\documentclass[11pt,letterpaper]{article}

\usepackage{amsfonts,amsmath,amssymb}
\usepackage{amsthm}
\usepackage[T1]{fontenc}
\usepackage{mathabx}
\usepackage[dvips]{graphicx}
\usepackage[usenames,dvipsnames]{xcolor}

\usepackage{CJKutf8}
\usepackage{enumerate}
\usepackage[round]{natbib}
\usepackage{multirow}
\usepackage{caption}
\usepackage{url}
\newcommand{\ER}{Erd\H{o}s-R\'{e}nyi }

\newcommand{\qt}{q_t}

\newcommand{\qtm}{q_{t-1}}
\newcommand{\qast}{q^{\ast}}

\newcommand{\J}{\mathcal{J}}

\newcommand{\vect}[1]{\mbox{\boldmath $#1$}}
\allowdisplaybreaks[1]

\newtheorem{definition}{Definition}
\newtheorem{assumption}{Assumption}
\newtheorem{proposition}{Proposition}
\newtheorem{theorem}{Theorem}
\newtheorem{corollary}{Corollary}

\newcommand*\rot{\rotatebox{90}}

\newcommand{\new}[1]{{\color{black}  {#1}}}
\newcommand{\add}[1]{{\color{black}  {#1}}}

\newcommand{\addr}[1]{{\color{black}  {#1}}}

\title{Dynamics of diffusion on monoplex and multiplex networks: \\ A message-passing approach\footnote{
We would like to thank Takehisa Hasegawa, Ryoji Sawa and Takashi Shimizu for their useful comments and suggestions. 
Kobayashi acknowledges financial support from JSPS KAKENHI\ 19H01506, 20H05633 and 22H00827. Onaga acknowledges financial support from JSPS KAKENHI 19K14618 and 19H01506.}}

\date{\today}
\author{Teruyoshi Kobayashi\thanks{kobayashi@econ.kobe-u.ac.jp} \\
Kobe University \and Tomokatsu Onaga\thanks{onaga@se.is.tohoku.ac.jp}\\ Tohoku University}

\begin{document}

\begin{CJK*}{UTF8}{gbsn}

\maketitle

\vspace{-10mm}
\begin{abstract} \setlength{\baselineskip}{13pt}
 New ideas and technologies adopted by a small number of individuals occasionally spread globally through a complex web of social ties. Here, we present a simple and general approximation method, namely, a message-passing approach, that allows us to describe the diffusion processes on \add{(sparse) random} networks in an almost exact manner. We consider two classes of binary-action games where the best pure strategies for individual players are characterized as variants of the threshold rule. We verify that the dynamics of diffusion observed on synthetic networks are accurately replicated by the message-passing equation, whose fixed point corresponds to a Nash equilibrium, while the conventional mean-field method tends to overestimate the size and frequency of diffusion. Generalized cascade conditions under which a global diffusion can occur are also provided. We extend the framework to analyze multiplex networks in which social interactions take place in multiple layers.

{\flushleft JEL \emph{classification:} C72, D85, L14}
\end{abstract}

\section{Introduction}
 
Decisions are often a result of the influence of others. One would be more likely to adopt a new technology if many friends and colleagues are already using it, while one may forego the technology if only a few have adopted it.  
Such peer effects through social ties are modeled as coordination games on a social network~\citep{kandori1993ECMAlearning,morris2000contagion,Jackson2008book}. In a class of $2\times 2$ coordination games, in which there are two pure Nash equilibria, changes in the strategies undertaken by a small fraction of players may initiate contagion that finally leads all the players in the network to change their behavior~\citep{morris2000contagion,jackson2007diffusion,lopez2008GEBdiffusion,galeotti2010network,lopez2012GEBinfluence,lelarge2012diffusion,sadler2020diffusion,melo2021uniqueness}. 
\add{An alternative} framework for studying strategic interactions in networks is provided by \cite{ballester2006}. 
In their model, a player takes a real-valued action $x\geq 0$ to maximize a quadratic utility function that depends on the neighbors' actions, as well the player's own action. 
\addr{
They show that players' best strategies in equilibrium are characterized by their positions in the network, where the best action $x$ is proportional to the player's Bonacich centrality\footnote{See, \cite{chen2018AEJmultiple} for an extension of the \citeauthor{ballester2006}'s \citeyearpar{ballester2006} model.}.
}

 In \add{the standard} coordination games on a network, the best pure strategy is formulated as a \emph{fractional-threshold rule} in which a player becomes ``active'' if a certain fraction of the neighbors are active~\citep{morris2000contagion,Jackson2008book}.  In \new{the utility-maximizing games} proposed by \cite{ballester2006}, it can also be shown that the best strategy is described as a threshold rule as long as players' actions are binary, i.e., $x\in\{0,1\}$, but an essential difference is that the best binary action is expressed as an \emph{absolute-threshold rule}; a player decides to be active (i.e., select $x=1$) if the total number, as opposed to the fraction, of active neighbors exceeds a certain threshold, as in the classical threshold model of \cite{Granovetter1978}. Thus, \new{the two types of games} fall into two distinct classes of threshold models.

 While the behavior of each player is determined by a threshold rule at the local level, understanding the aggregate dynamics of diffusion and its equilibrium property is a non-trivial problem if players are interconnected in a complex manner.
As a conventional approach for achieving this goal, mean-field approximations have been extensively used to calculate the steady-state fraction of active players~\citep{jackson2007diffusion,lopez2006contagionIJGT,lopez2008GEBdiffusion,lopez2012GEBinfluence,lelarge2012diffusion}. A key idea of the mean-field method is that the (possibly heterogeneous) probabilities of neighbors being active are approximated by a constant probability that a randomly chosen neighbor \addr{is} active. 

 In this study, we develop a more accurate approximation method that allows us to describe the dynamics, as well as the equilibrium, of contagion on complex networks in an almost exact manner. Our method, called the \emph{message-passing approach}, is a more elaborated version of the conventional mean-field approximation in that the message-passing approach takes into account the directionality of the spreading process~\citep{gleeson2018message,dall2021coordination}\footnote{
 Message-passing approaches are approximation methods that have been developed in statistical physics. 
 In the study of ferromagnetism, Ising model, in which each of atomic spins is in one of two states $\{-1,+1\}$, was initially analyzed based on the mean-field theory developed by \cite{weiss1907}. A more accurate solution of Ising model was obtained using Bethe approximation~\citep{bethe1935statistical}, and then a variant of the message-passing approximation, called the cavity method, was developed as an extended version of Bethe approximation, which can be applied to wider classes of models in statistical physics~\citep{mezard1987spin}. 
 }.
 In the mean-field method, the probability of a neighboring player being active is given as a function of the probability that the neighbor's neighbors are active, and a fixed point of the self-consistent equation corresponds to the steady state of the spreading process, \addr{i.e., a Nash equilibrium}~\citep{jackson2007diffusion,Jackson2008book,jackson-zenou2015games}.
 However, the recursive expression of the self-consistent equation necessarily incorporates a repercussion of peer effects among neighboring players because social influence, or ``messages,'' may be transmitted multiple times between the neighbors. The mean-field approximations \addr{may thus overestimate the activation probability of neighbors}.
 The message-passing method overcomes this problem by imposing a directionality condition that one of the neighbors to which a message will be passed at time step $t+1$ is not yet active at time step $t$. 
 In general, mean-field methods are accurate enough in the limit of large degrees, where network density is sufficiently high, but they do not necessarily provide good approximations for sparse networks~\citep{Dorogovtsev2008RevModPhys,gleeson2011high}. 
 Nevertheless, in the previous studies on network games, little quantitative validation has been performed to examine if the mean-field approximation correctly captures the ``true'' Nash equilibrium.

 The main results of this study are summarized as follows. First, for each class of network games, we obtain a generalized version of the conventional \emph{cascade conditions}. Several studies provide analytical conditions that predict a parameter space in which a vanishingly small fraction of active players can cause a global cascade~\citep{Watts2002,Gleeson2007,lopez2008GEBdiffusion,lelarge2012diffusion,dall2021coordination}. However, if the share of initially active players, called ``seed players,'' is positive, which seems to be the case in real social networks, the accuracy of these cascade conditions would be undermined because the fraction of active players is assumed to be zero in evaluating the first derivative.
  Based on the message-passing equation, we derive generalized cascade conditions that incorporate the previously proposed ones as special cases.
 We show that the parameter space indicated by the generalized cascade conditions almost exactly matches the ``ground-truth'' cascade region obtained by numerical simulations on synthetic networks.

Second, the message-passing method can accurately describe the aggregate dynamics of contagion, i.e., the dynamic path of the share of active players to the equilibrium level, while previous studies on network games focus mainly on equilibrium property. The message-passing equation can be used to explain how fast and to what extent a ``message'' will spread through social ties, as well as how many players will finally receive the message. We show that the share of active players at a given time step $t$ is approximated by the probability of a randomly chosen player being active at time $t$, which is calculated by iterating the message-passing equation $t$ times.

 Finally, we extend the baseline game-theoretic framework for single networks (i.e., \emph{monoplex} networks) to the case of \emph{multiplex} networks. A multiplex network is a set of layers in each of which a network is formed by a common set of nodes~\citep{Brummitt2012_PRER,Brummitt2015PRE,bianconi2018book}.
 For instance, students' friendship in a school shapes a school social network, while students using Twitter may also be connected through follower-followee relationships.
 In \addr{such cases}, a monoplex model will not be enough to study the aggregate dynamics of social influence. If Twitter followers are more influential than school friends, or vice versa, then we need to take into account the inter-layer heterogeneity that would be reflected as a difference in the payoff parameters. 
 To analyze the contagion of behavior in the framework of multiplex networks, we redefine the best strategies in the two classes of games and develop a generalized cascade condition based on the largest eigenvalue of the Jacobian matrix that has a derivative of the multivariate message-passing equation in each element. We reveal that, for a given connectivity of networks, inter-layer heterogeneity tends to enhance contagion because players are more likely to become active in the influential layer, which would trigger a global contagion.


Our analysis is located at the intersection of game theory and network science. 
It has been widely recognized that the role of networks (e.g., social ties between individuals and trading relationships between firms) is crucial in understanding social and economic phenomena~\citep[c.f.,][]{amir2021oligopoly,barbieri2021complementarity,luo2021network,masatlioglu2021decision,meng2021competitive,la2022geographical}.
Game theorists have \addr{long} studied the diffusion of behavior on networks in which strategic interactions take place between players connected by social ties~\citep{Jackson2008book,Easley2010book,jackson2011overview,jackson-zenou2015games,tabasso2019diffusion}. A seminal work by \cite{morris2000contagion} examines a class of $2\times 2$ coordination games on regular graphs and defines a contagion threshold of the payoff parameter, showing how the possibility of contagion depends on the network structure.
Another type of model to study network games is based on a utility function that depends on neighbors' actions. \cite{ballester2006} consider a continuous-action game on networks and show that the optimal strategy of a player in a Nash equilibrium is proportional to the player's Bonacich centrality. \cite{chen2018AEJmultiple} extend the framework to accommodate multiple (or bilingual) actions, as examined by \cite{goyal1997non}, \cite{immorlica2007role}, and \cite{oyama2015bilingual} in the context of coordination games.


 In network science, the collective dynamics of individuals' behavior have been recognized as an important research topic since the early 2000's~\citep{Watts2002,Watts2007}. 
 Watts' model of threshold cascades has been extended to a wide variety of collective phenomena, such as information cascades~\citep{Nematzadeh2014,kobayashi2015trend,unicomb2021dynamics}, and default contagion in financial networks~\citep{GaiKapadia2010,Cont2013,hurd2016Book,caccioli2018review}. The framework has also been extended to analyze contagion on multiplex networks in which nodes belong to multiple layers~\citep{Brummitt2012_PRER,Yagan2012,Brummitt2015PRE,bianconi2018book,unicomb2019reentrant}. 
 
 These studies in the field of network science usually take the threshold rules as given, but our work provides a microfoundation from a game-theoretic perspective; in both fractional and absolute threshold models, the threshold value is obtained as a function of the payoff parameters and the preference parameters. Many studies analyzing \addr{the} fractional-threshold models employ variants of the message-passing equation proposed by \cite{Gleeson2007} and \cite{Gleeson2008} to calculate the steady-state equilibrium.
 While the message-passing approach is not new in network science~\citep{dall2021coordination}, to the best of our knowledge, we are the first to provide formal proofs for the existence of and convergence to a fixed point of the message-passing equation.

\section{Binary-action games on monoplex networks}

\subsection{\add{Random network and locally tree-like structure}}

We first consider a single (i.e., \emph{monoplex}) uniformly random network formed by a sufficiently large number of players. Player $i$ is connected with $k_i$ other players by undirected and unweighted edges. $k_i$ is called the \emph{degree} of player $i$ (or node $i$).
Players at the end of the edges emanating from $i$ are called \emph{neighbors} of player $i$. At each time step, each pair of neighbors plays a one-to-one game.
The only assumption we make about the network is that it has a locally tree-like structure, i.e., there are no cycles (e.g., triangles, quadrangles) at least locally. 
This indicates that neighbors of a node are not directly connected, so the neighbors are not coordinating with each other.
Examples of this class of networks include sparse \ER random graphs
and the configuration models in which the degree distribution is prespecified while the nodes are connected at random subject to the degree constraint~\citep{molloy1995critical,newman2018book2nd}. 

In \ER networks and the configuration models, the clustering coefficient $C$, which is the average probability that two neighbors of a node are connected, is generally given by\footnote{See \cite{newman2018book2nd}.
\addr{For \ER networks with $\langle k\rangle \approx pN$, where $p$ denotes a connecting probability,} the clustering coefficient reduces to $C = \langle k\rangle/N$.}
\begin{align}
    C = \frac{1}{N}\frac{(\langle k^2\rangle - \langle k\rangle)^2}{\langle k \rangle^3},
\end{align}
where $N$ is the number of nodes in the network and $\langle \cdot \rangle$ denotes the average operator. Note that the moments $\langle k\rangle$ and $\langle k^2\rangle$ are fixed in \ER and the configuration models since the degree distribution is prespecified.
We thus have $\lim_{N\to\infty}C= 0$, meaning that the fraction of local triangles among all pairwise combinations is at most $O(N^{-1})$, which is negligible when the network is sufficiently large. 
\add{The density of these networks, given by $2M/(N(N-1))=\langle k \rangle/(N-1)$ where $M$ denotes the number of edges, is also vanishing as $N\to \infty$, meaning that they are sparse.}

\add{\subsection{Fractional threshold model}}

\begin{table}[tbh]
    \centering
        \caption{Payoffs in a coordination game. $(u_i,u_j)$ denotes the combination of payoffs, where $u_i$ (resp. $u_j$) is the payoff to player $i$ (resp. $j$). $a,c>0.$}
    \begin{tabular}{cccc}
   & &\multicolumn{2}{c}{Player $j$}\\
    \cline{2-4}
     \multicolumn{1}{c|}{}   & &  \multicolumn{1}{|c}{0}    & \multicolumn{1}{|c|}{1}  \\
    \cline{2-4}
    \multicolumn{1}{c|}{}&    0 &\multicolumn{1}{|c}{$(0,0)$} & \multicolumn{1}{|c|}{$(0,-c)$}\\
    \cline{2-4}
     \multirow{2}*{\rot{\rlap{~Player $i$}}}
     &  \multicolumn{1}{|c}{1} &\multicolumn{1}{|c}{$(-c,0)$} & \multicolumn{1}{|c|}{$(a,a)$}\\
    \cline{2-4}
    \end{tabular}
    \label{tab:one-good game}
\end{table}

We first describe a \add{standard} coordination game on a network. Each player takes action 0 or action 1, the payoffs of which are given in Table~\ref{tab:one-good game} where $a,c>0$.
Throughout the analysis, players taking action 1 are called \emph{active}, while those taking action 0 are called \emph{inactive}. 


Now let us consider a situation in which an arbitrary chosen player changed the strategy from action 0 to action 1. We call this player a \emph{seed}. If the seed player adopts action 1, there arises a possibility that some of the neighbors may accordingly change their strategy, because they might be better off if they would coordinate with the seed. It is assumed that seed players will never revert their strategy. 
The condition for player $i$ to take action 1 is given by $-c(k_i-m_i) + a m_i>0$, or
\begin{align}
    \frac{m_i}{k_i} > \frac{c}{a+c}\equiv \phi,
    \label{eq:fractional_threshold}
\end{align}
where $m_i$ denotes the number of player $i$'s neighbors that take action 1. If a player takes action 1, it would affect the neighbors through the threshold condition \eqref{eq:fractional_threshold}, which may cause further changes in the neighbors' neighbors' action, and so forth~\citep{morris2000contagion}. 
\addr{
In this way, a change in the behavior of a single seed player in an infinitely large network
may cause a \emph{global cascade} of behavioral changes~\citep{Watts2002}.
}
\add{While the original Watts model assumes that the initial seed fraction is vanishingly small, here we will generalize the framework so that the effect of a positive fraction of seed nodes can be formally analyzed.}

It should be noted that players' actions are irreversible in the sense that players \addr{adopting action~1 will not have an incentive to revert to action 0}\footnote{Recall that seed nodes are \emph{assumed} to be irreversible.}.
This guarantees the monotonicity of the contagion process, and thus there is at least one stationary equilibrium~\citep{morris2000contagion}.
We call a class of contagion models in which the threshold condition is given by  \eqref{eq:fractional_threshold} \emph{fractional-threshold models}~\citep{Watts2002,Karimi2013PhysicaA}\footnote{\cite{unicomb2021dynamics} call this type of cascade model a \emph{relative} threshold model.}.
\begin{definition}
If players in a network change their strategies according to \eqref{eq:fractional_threshold}, the resulting cascading behavior is called a fractional-threshold contagion.
\end{definition}

\new{\subsection{Absolute threshold model}}

There is a strand of literature on continuous-action games on networks in which each player takes an action represented by a real value $x\geq 0$~\citep{ballester2006,jackson-zenou2015games}. Typically, player $i$ maximizes the following quadratic utility function
\begin{align}
    u_i(x_i;{\bf{x}}_{-i}) = \alpha x_i - \frac{1}{2}x_i^2 +\gamma\sum_{j\ne i} \mathcal{A}_{ij}x_ix_j, 
\end{align}
where ${\bf{x}}_{-i}$ denotes the vector of actions taken by players other than $i$, and $\mathcal{A}_{ij}$ is the $(i,j)$th element of the adjacency matrix $\mathcal{A}$; $\mathcal{A}_{ij}=1$ if there is an edge between players $i$ and $j$, and $\mathcal{A}_{ij}=0$ otherwise. Since there is no self-loop and edges are undirected, we have $\mathcal{A}_{ii}=0$ and $\mathcal{A}_{ij}=\mathcal{A}_{ji}$. The first term, $\alpha x_i$, captures the payoff of taking action $x_i$ with weight $\alpha>0$, while the second term, $x_i^2/2$, represents the cost of taking action. 
 The third term, $\gamma\sum_{j\ne i} \mathcal{A}_{ij}x_ix_j$, denotes the benefit of coordinating with neighbors.
We focus on the case of $\gamma \geq 0$ to investigate the dynamics of a monotonic cascading process, in which a stationary equilibrium is guaranteed.

Now we consider a case where players' actions take binary values: $x_i\in\{0,1\}$. The utilities of taking action $x_i=0$ and $x_i=1$ are respectively given as
\begin{align}
    u_i(0;{\bf{x}}_{-i}) & = 0, \\
    u_i(1;{\bf{x}}_{-i}) & = \alpha - \frac{1}{2}+\gamma m_i. 
\end{align}
\addr{Note} that if $x_i=0$ is the status quo for all players, this \addr{condition} is similar to \addr{the one in} the fractional threshold model in that a player's action may cause other players' actions to change.
 Specifically, player $i$ takes action 1 if $u_i(0;{\bf{x}}_{-i})$ < $u_i(1;{\bf{x}}_{-i})$, which is rewritten as
\begin{align}
    m_i > \frac{1}{\gamma}\left( \frac{1}{2}-\alpha\right) \equiv \theta.
    \label{eq:absolute_threshold}
\end{align}
The most important difference from the threshold condition in \new{the fractional model}, Eq.~\eqref{eq:fractional_threshold}, is that condition \eqref{eq:absolute_threshold} is independent of the degree $k_i$. A player will be activated if the number of active neighbors exceeds a certain threshold $\theta$.
We call this class of contagion models \emph{asolute-threshold models}~\citep{Granovetter1978,Karimi2013PhysicaA,unicomb2021dynamics}.
\begin{definition}
If players in a network change their strategies according to \eqref{eq:absolute_threshold}, the resulting cascading behavior is called an absolute-threshold contagion.
\end{definition}

\section{Analysis of contagion dynamics}

As discussed above, \addr{the best pure strategies in the two types of games fall into two different classes: fractional and absolute threshold rules}. In this section, we show how the contagion dynamics in these two classes of games can be analyzed within a unified framework.

\subsection{A message-passing approach}

 To study the dynamics of cascading behavior in an analytically tractable and intuitive manner, we employ a message-passing approach. 
 Below we present a message-passing equation that allows us to describe the dynamics of contagion in both classes of network games. We further propose analytical conditions with which we can identify under what circumstances a global cascade can take place.

\subsubsection{Message-passing equation}

 \begin{figure}[tb]
     \centering
     \includegraphics[width=6cm]{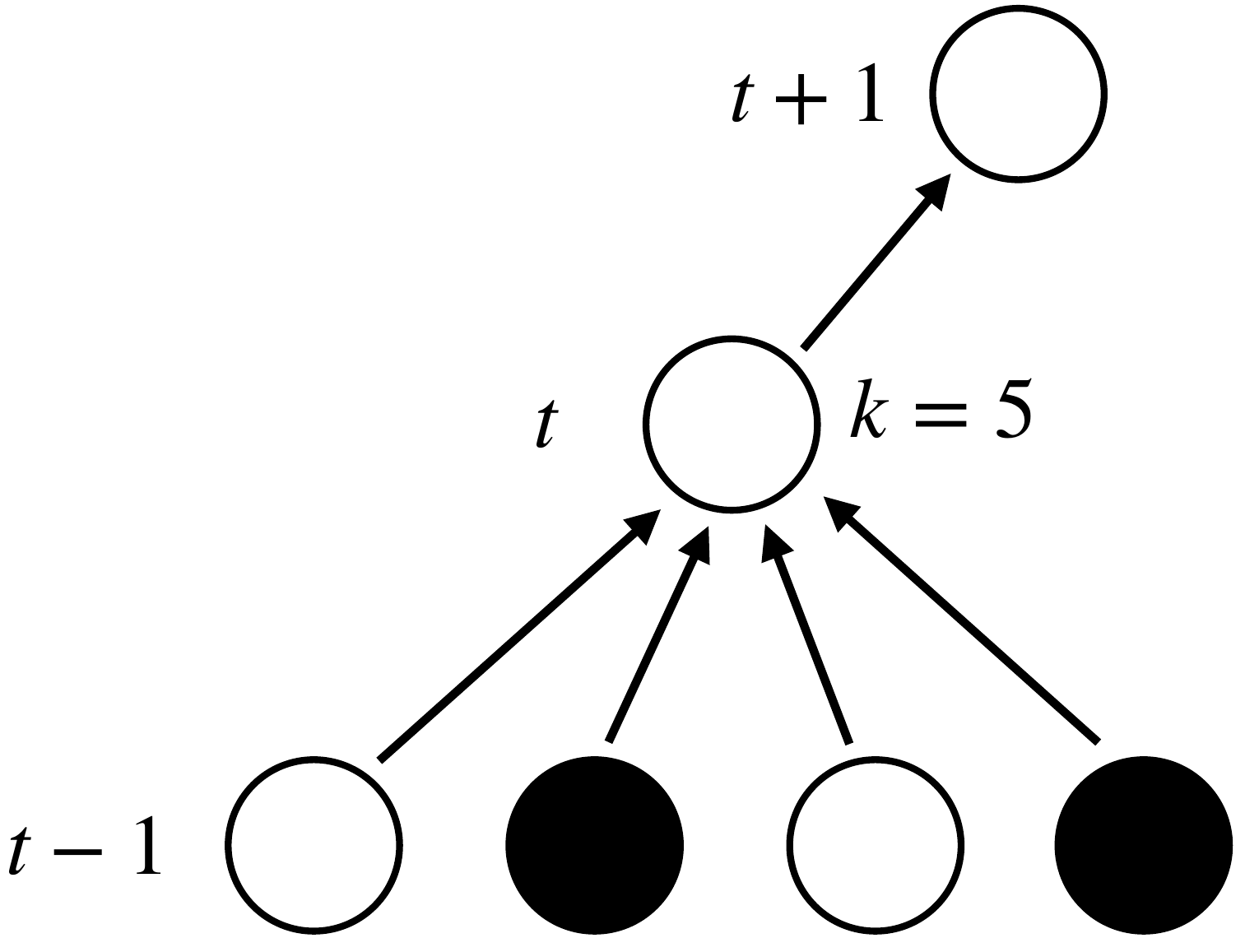}
     \caption{Schematic of the message-passing method.}
     \label{fig:schematic}
 \end{figure}
 
Let $\rho_0\in[0,1)$ denote the fraction of seed players among all players. That is, $\lfloor\rho_0 N\rfloor$ seed players are randomly selected from the population.  
 Let $\qt\in[0,1]$ be the probability that a randomly selected neighbor of a node is active at time step $t$. In the message-passing approach, the states of $k-1$ neighbors out of $k$ neighbors are taken into account in calculating the probability $\qt$, assuming that one neighbor that is not included in the $k-1$ neighbors is still inactive (Fig.~\ref{fig:schematic}). This allows us to respect the directionality in the spreading of influence; a player would become active after being influenced by (or receiving ``messages'' from) the $k-1$ neighbors, and the influence (or the ``message'') is passed on to a neighbor who is not yet active. Assuming that the network is locally tree-like, we calculate $\qt$ by iterating the following recursion equation (i.e., the message-passing equation):
 \begin{align} 
 \qt &= \rho_0 + (1-\rho_0)\sum_{k=1}^{\infty}\frac{k}{z}p_k \sum_{m=0}^{k-1} \binom{k-1}{m}  \qtm^m (1-\qtm)^{k-1-m}F\left( m,k \right),\notag 
\\
 &\equiv G(q_{t-1}),  \;\; \text{for }t=1,2,\ldots,
 \label{eq:recursion_q}
 \end{align}
 where $\binom{k-1}{m}\equiv (k-1)!/[k!(k-1-m)!]$, and $p_k$ denotes the degree distribution, i.e., the probability that a randomly chosen player has exactly $k$ neighbors.
 $z\equiv \langle k\rangle=\sum_{k=0}^\infty k p_k$ is the mean degree of the network, so $\frac{k}{z}p_{k}$ represents the probability that a randomly chosen neighbor of a player has degree $k$.\footnote{\add{If the network is directed, the summation in Eq.~\eqref{eq:recursion_q} should be taken from $m=0$ to $k$ since each edge has an intrinsic direction~\citep{Brummitt2015PRE}.
 \cite{lopez2012GEBinfluence} studies the impact that the presence of in/out-degree correlations (i.e., assortativity) would have on diffusion. 
 }}
 $F$ is a real-valued function bounded on $[0,1]$, which is called the \emph{response function}. We impose the following assumption on $F$:
 \begin{assumption}
 $F:\mathbb{Z}_{\geq 0}\times\mathbb{Z}_{>0}\to [0,1]$ has the following properties: \\
  For $k= 1,2,\ldots$,
 \begin{enumerate}[(i)]
     \item $F(0,k)=0$,
     \item  $F(m,k)$ is \addr{increasing} in $m$, i.e., $F(m_1,k)\leq F(m_2,k)$ for all $m_1,m_2\in \mathbb{Z}_{\geq 0}$ such that $m_1\leq m_2$.
 \end{enumerate}
\label{ass:F}
 \end{assumption}
 Part (i) of the assumption indicates that players will never get activated if there are no active neighbors, while part (ii) ensures that players are more likely to be active as they have more active neighbors for a given number of neighbors.

 Specifically, the response function in \new{the fractional threshold model} is defined as
 \begin{align}
 F(m,k) \equiv {F}^{\rm frac}(m,k) =
     \begin{cases}
      1 \;\; \text{ if }\; \frac{m}{k}>\phi, \\ 
      0 \;\; \text{ otherwise}.
     \end{cases}
     \label{eq:fraction_rule}
\end{align}
Clearly, if $\phi\geq 0$, the response function satisfies Assumption~\ref{ass:F}. 
In \new{the absolute threshold model}, the response function is given by
\begin{align}
 F(m,k) \equiv {F}^{\rm abs}(m) = 
     \begin{cases}
      1 \;\; \text{ if }\; m >\theta, \\
      0 \;\; \text{ otherwise}.
     \end{cases}
     \label{eq:absolute_rule}
 \end{align}
Again, if $\theta\geq 0$, this response function satisfies Assumption~\ref{ass:F}.

 For a given $q_t$, we can calculate the average fraction of active players at time $t$, denoted by $\rho_t$, as
\begin{align}
 \rho_t = \rho_0 + (1-\rho_0)\sum_{k=1}^{\infty}p_k \sum_{m=0}^k \mathcal{B}_m^k(q_t)F\left( m,k \right),  \label{eq:rhot}  
\end{align} 
 where $\mathcal{B}_m^k(q)\equiv \binom{k}{m}  q^m (1-q)^{k-m}$ denotes a binomial distribution with parameter $q$.

\subsubsection{Relation to a more general approximation method}

\cite{gleeson2011high,gleeson2013binary} argues that a general approximation method based on \emph{approximate master equations} (AMEs) is far more accurate than the mean-field approximations in a wide variety of contagion models. Let $s_{k,m}(t)$ denote the fraction of $k$-degree nodes that are inactive and have $m$ active neighbors at time $t$, where the average fraction of active nodes is given by $\rho(t) = 1- \sum_k p_k\sum_{m=0}^k s_{k,m}(t)$. Using the AME formalism, we can express the dynamics of $s_{k,m}$ in continuous time as follows:
\begin{align}
    \frac{d}{dt} s_{k,m} = -F(m,k)s_{k,m} - \beta\cdot (k-m)s_{k,m} + \beta\cdot(k-m+1)s_{k,m-1},
    \label{eq:AME}
\end{align}
where $\beta \equiv [\sum_k p_k\sum_{m=0}^k (k-m)F(m,k)s_{k,m}]/[\sum_k p_k\sum_{m=0}^k (k-m)s_{k,m}]$ denotes the average transition rate at which a neighbor of an inactive node becomes active. In the AME method, we need to consider three factors that can change $s_{k,m}$. The first term on the RHS of \eqref{eq:AME} represents the fraction of $k$-degree nodes having $m$ active neighbors that newly change their state from inactive to active. The second term represents the fraction of nodes that leave the $s_{k,m}$ class because one of the $(k-m)$ inactive neighbors is newly activated and the number of active neighbors becomes $m+1$. The third term represents the fraction of nodes that newly enter the $s_{k,m}$ class from the $s_{k,m-1}$ class because the number of inactive neighbors changes from $k-m+1$ to $k-m$. \cite{gleeson2011high,gleeson2013binary} shows that a wide variety of contagion processes on networks can be described by the AME method almost exactly. The downside of the AME method, on the other hand, is that since the total number of combinations $(k,m)$ is $(k_{\rm max}+1)(k_{\rm max}+2)/2$, the number of equations grows quadratically with the maximum degree $k_{\rm max}$. Thus, solving the system of differential equations in the AME approach can be computationally expensive for well connected networks.  

Despite the increased number of equations to be solved in the AME formalism, it is shown that the message-passing equation~\eqref{eq:recursion_q} is directly derived from the system of AMEs as long as the threshold models of type \eqref{eq:fraction_rule} and \eqref{eq:absolute_rule} are considered. We summarize the equivalence between the message-passing and the AME methods in the following proposition:
\begin{proposition}[\citealt{gleeson2013binary}]
Suppose that response function $F$ is given by \eqref{eq:fraction_rule} or \eqref{eq:absolute_rule}. Given the initial fraction of active nodes at time $t=0$, denoted by $\rho(0)$, the system of differential equations in the AME formalism given by \eqref{eq:AME}, for $m =0,\ldots,k$ and $k=0,\ldots,k_{\rm max}$, reduces to the following system of two differential equations:
\begin{align}
    \frac{d}{dt}\rho &= v(q) - \rho, \label{eq:diff_mp}\\
    \frac{d}{dt} q &= g(q) - q,  \\
    \text{where } \hspace{2cm}\notag \\
      v(q) = \rho(0) &+ [1-\rho(0)]\sum_k p_k\sum_{m=0}^k \mathcal{B}_{m}^{k}(q)F(m,k), \\
    g(q) = \rho(0) &+ [1-\rho(0)]\sum_k \frac{kp_k}{z}\sum_{m=0}^{k-1} \mathcal{B}_{m}^{k-1}(q)F(m,k). \label{eq:AME_MP_derivation}
\end{align}
\end{proposition}
\begin{proof}
 See the section VII and Appendix F of \cite{gleeson2013binary}.
\end{proof}
 Eqs.~\eqref{eq:diff_mp} -- \eqref{eq:AME_MP_derivation} indicate that the dynamics captured by the message passing equation \eqref{eq:recursion_q} are indeed equivalent to that described by the AMEs in the class of threshold models that we consider. This suggests that while the message-passing method is apparently similar to the conventional mean-field approximation in that the mean probability $q$ is the only variable in Eq.~\eqref{eq:recursion_q}, the message-passing equation essentially incorporates the transition of neighbors' states that is neglected in the conventional mean-field method.

 \subsection{Properties of the message-passing equation~\eqref{eq:recursion_q}}
 \subsubsection{Monotonicity}
 
  Given that the response function $F$ satisfies Assumption~\ref{ass:F}, we now characterize the function $G$ on $[0,1]$.  
  We first describe the boundary conditions of $G$.
 \begin{proposition}
 Let $G$ on $[0,1]$ be a real-valued function defined by \eqref{eq:recursion_q}. If $F:\mathbb{Z}_{\geq 0}\times\mathbb{Z}_{>0}\to [0,1]$ satisfies Assumption~\ref{ass:F} (i), then $G(0)=\rho_0$ and $G(1) \leq 1$ for $\rho_0\in [0,1)$.
  \label{prop:G_01}
 \end{proposition}
 \begin{proof}
  From the convention $0!=1$ and Assumption~\ref{ass:F} (i), it is straightforward to show that $G(0) = \rho_0 + (1-\rho_0)\sum_{k=1}^{\infty} \frac{k p_{k}}{z}\binom{k-1}{0}F(0,k) = \rho_0$. We also have $G(1) = \rho_0 + (1-\rho_0)\sum_{k=1}^{\infty} \frac{k p_{k}}{z}F(k-1,k)$. Since $\sum_{k=1}^{\infty} \frac{k p_{k}}{z}=1$ and $0\leq F\leq 1$, it is clear that $G(1)\leq 1$.
   The equality is attained if there is $\underline{k}>1$ such that $\sum_{k=\underline{k}}^{\infty} \frac{k p_{k}}{z}F(k-1,k)=1$. 
 \end{proof}
 The following proposition states that $G$ is monotone.
 \begin{proposition}
  If $F:\mathbb{Z}_{\geq 0}\times\mathbb{Z}_{>0} \to [0,1]$ satisfies Assumption~\ref{ass:F}~(ii), then $G$ is \addr{increasing} on $[0,1]$. 
  \label{prop:G_monotonicity}
 \end{proposition}
\begin{proof}
 See, \ref{sec:proof_G_mononone}.
\end{proof}
 Propositions \ref{prop:G_01} and \ref{prop:G_monotonicity} can be summarized in the following corollary.
 \begin{corollary}\label{col:G_01}
  Let $F:\mathbb{Z}_{\geq 0}\times\mathbb{Z}_{>0} \to [0,1]$ satisfy Assumption~\ref{ass:F}. Then, $G:[0,1]\to[0,1]$ is a monotone map.
  \end{corollary}

 \subsubsection{Convergence}
 \addr{
 We now examine convergence of the sequence $\{q_t\}_{t\in \mathbb{Z}_{\geq 0}}$ generated by iterating the recursion equation~\eqref{eq:recursion_q} from an initial value $q_0\in[0,\rho_0]$.
 To this end, we use ``Kleene's fixed point theorem''~\citep{baranga1991contraction,stoltenberg1994mathematical,kamihigashi2015application}.
 Let $(P,\leq)$ be a partially ordered set. 
 $(P,\leq)$ is \emph{$\omega$-complete} if every increasing sequence $\{x_n\}_{n\in \mathbb{Z}_{\geq 0}}$ in $P$ such that $x_n\leq x_{n+1}$ has a supremum in $P$.
 A function $f:P\to P$ is \emph{$\omega$-continuous} if $f\left(\sup_{n\in\mathbb{Z}_{\geq 0}}x_n\right) = \sup_{n\in\mathbb{Z}_{\geq 0}}f(x_n)$ holds for every increasing sequence $\{x_n\}_{n\in \mathbb{Z}_{\geq 0}}$ in $P$ having a supremum\footnote{It is clear that an $\omega$-continuous function $f$ is increasing: $f(x)\leq f(y)$ whenever $x\leq y$.}.
 \begin{theorem}[Kleene's fixed point theorem]
  Let $(P,\leq)$ be an \emph{$\omega$-complete} partially ordered set, and let $f:P\to P$ be an \emph{$\omega$-continuous} function.
  If $x\in P$ is such that $x\leq f(x)$, 
  then $x^\ast\equiv \sup_{n\in\mathbb{Z}\geq 0}f^n(x)$ is the least fixed point of $f$ in $P$.
  \label{th:kleene}
 \end{theorem}
 \begin{proof}
  See \citet[p.~24]{stoltenberg1994mathematical}.
 \end{proof}
 Using Kleene's fixed point theorem, we show that the sequence $\{q_t\}_{t\in \mathbb{Z}_{\geq 0}}$ generated by the message-passing equation~\eqref{eq:recursion_q} for a given $q_0\in[0,\rho_0]$ converges to the least fixed point of $G$. 
 \begin{theorem}
 Let $F:\mathbb{Z}_{\geq 0}\times\mathbb{Z}_{>0} \to [0,1]$ satisfy Assumption~\ref{ass:F}. 
 Then for any initial value $q_0\in [0, \rho_0]$ with $0\leq \rho_0< 1$, 
 the sequence $\{q_t\}_{t\in\mathbb{Z}_{\geq 0}}$ generated by the message-passing equation \eqref{eq:recursion_q} converges to the least fixed point $q^* = G(q^*)$, where $q^{*}=\sup_{t\in\mathbb{Z}_{\geq 0}}G^t(q_0)\in[0,1]$.
\label{th:fixed_point}
\end{theorem}
\begin{proof}
 See \ref{sec:proof_fixed_point}.
\end{proof}
}


  Intuitively, the convergence of a monotonic sequence $(q_0,q_1,\ldots)$ to $q^*$ suggests that the influence of seed players spreads gradually through network edges, and that the probability of a randomly chosen neighbor being active will reach a steady-state equilibrium. 
  In fact, depending on the structure of the network, $q^*$ may not be the unique fixed point of $G$. To illustrate this, we use standard \ER random networks in which two nodes are connected with a constant probability $p$~\citep{Erdos1959PublMath}. 
  While this class of networks is not necessarily realistic, it has been extensively used in network analysis as an agnostic benchmark in which no structural information is presumed (c.f., \citealt{Watts2002,Jackson2008book,lelarge2012diffusion}). In particular, \ER networks have several desirable properties for the analysis of contagion: i) the average degree $z$ is the only parameter that describes the connectivity of an \ER network, so it is straightforward to quantify the effect of network connectivity on the spreading of contagion. ii) For a given finite average degree $z<\infty$, an \ER network becomes sparse as the number of nodes $N$ goes to infinity since $p\equiv z/(N-1)$. The sparsity ensures that the presence of local cycles and clusters is asymptotically negligible as $N\to \infty$,\footnote{See, \cite{newman2018book2nd}, ch.~11.} which is compatible with the assumption that the network is locally tree-like.   

\begin{figure}[tb]
     \centering
     \includegraphics[width=12cm]{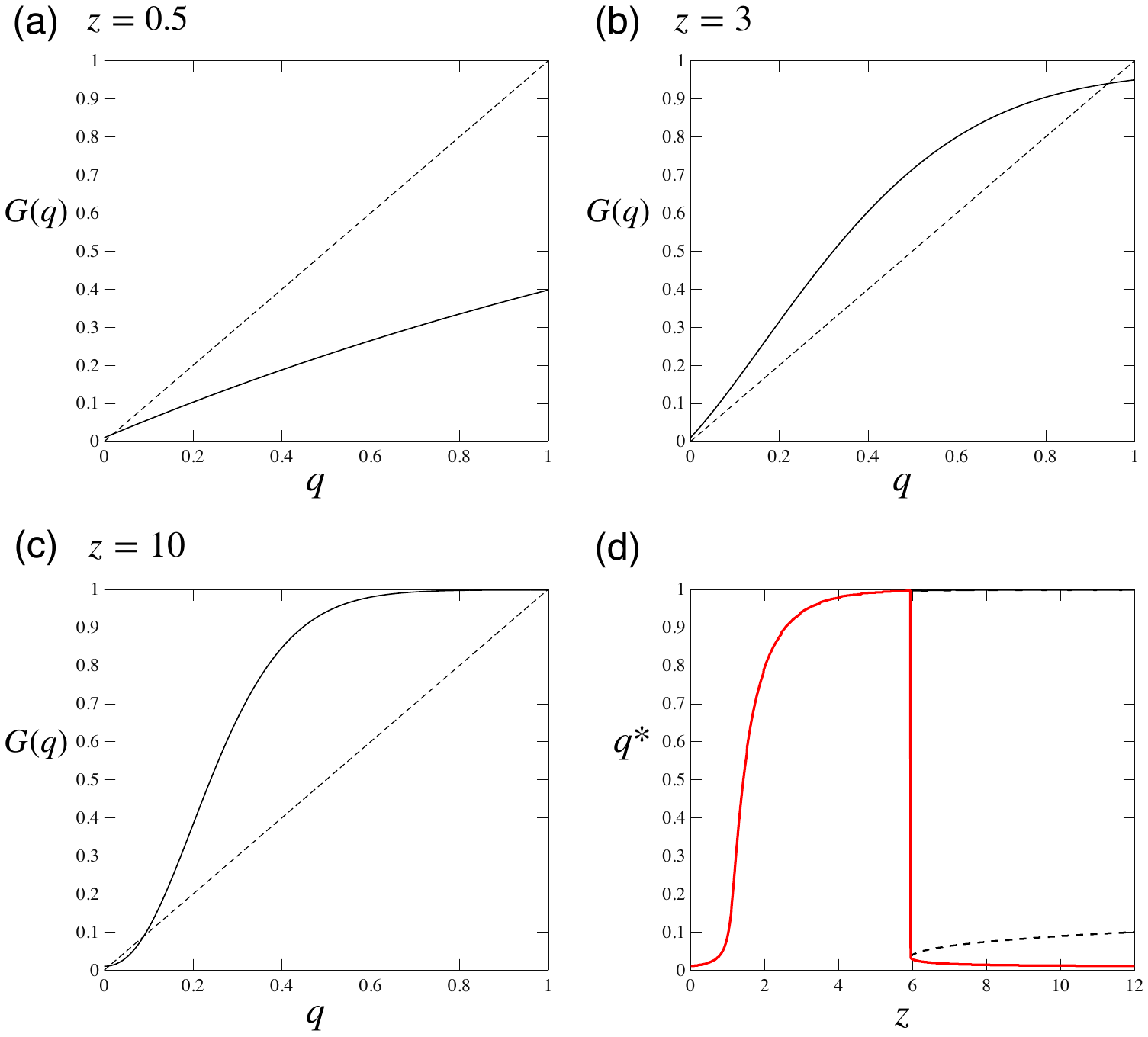}
     \caption{Fixed points of recursion equation \eqref{eq:recursion_q} for \new{the fractional threshold model}. We assume \ER random networks with $\rho_0=0.01$ and $\phi=0.2$. The mean degree $z$ is (a) $0.5$, (b) $3$, and (c) 10. Solid and dashed lines denote $G(q)$ and $45$-degree line, respectively.
     In panel (d), a bifurcation diagram of fixed points is shown. Black solid and dashed lines denote stable and unstable fixed points, respectively. Red solid line denotes the fixed point  $q^*$, defined in Theorem~\ref{th:fixed_point}. }
     \label{fig:bifurcation_Watts}
 \end{figure}
 
  Figs.~\ref{fig:bifurcation_Watts} and \ref{fig:bifurcation_Grano} in Online Appendix illustrate how the number of fixed points of $G$ depends on the average degree. 
  \add{In each panel of a--c in these figures, the curve (solid line) represents the function $G(q)$ for a given mean degree annotated at the top. 
  Note that the degree distribution $p_k$ is given by a Poisson distribution with mean $z$ for \ER networks, i.e., $p_k = z^ke^{-z}/k!$. 
  Note that $G(1)<1$ holds for networks with relatively small $z$ since we have $G(1) = \rho_0 + (1-\rho_0)\sum_{k=1}^{\infty} \frac{k p_{k}}{z}F(k-1,k)$ from the proof of Proposition~\ref{prop:G_01}.
  When $z$ is large enough, we would have a minimum degree $\underline{k}$ such that $p_k=0$ for all $k\in\{k:k<\underline{k}\}$ and $\sum_{k=\underline{k}}^{\infty} \frac{k p_{k}}{z}F(k-1,k)=1$, thereby leading to $G(1)=1$. However, if there are nodes having small $k$ such that $k<\underline{k}$ (i.e., $p_k>0$ for some $k \in \{k:k<\underline{k}\}$), then we have $G(1)<1$, in which case $q=1$ is not a fixed point of $G$. Indeed, the latter case corresponds to highly sparse networks in which the fraction of the largest connected component is less than one, so the influence of initial seeds would not reach isolated/disconnected components. 
  }

  The bifurcation diagram (Figs.~\ref{fig:bifurcation_Watts}d and \ref{fig:bifurcation_Grano}d) shows that the number of fixed points varies between 1 and 3 as the connectivity of network changes. 
  \addr{As shown in Theorem~\ref{th:fixed_point}, the fixed point $q^*$ obtained by iterating the recursion equation from $q_0\in [0,\rho_0]$ is the least fixed point of $G$ (Red solid line in Figs.~\ref{fig:bifurcation_Watts}d and \ref{fig:bifurcation_Grano}d).}
In equilibrium, it is optimal for all players not to change their strategies, i.e., a Nash equilibrium is attained for every pair, so no further contagion occurs. The steady-state probability of a randomly chosen player being active is thus given by
 \begin{align} 
 \rho^* = \rho_0 + (1-\rho_0)\sum_{k=1}^{\infty}p_k \sum_{m=0}^k \mathcal{B}_m^k(q^*)F\left( m,k \right).
 \label{eq:rho_ss}
 \end{align}
 If there are a sufficiently large number of players, $\rho^*$ can be interpreted as the share of active players in a Nash equilibrium.
  \add{Now we define \emph{global cascade} as follows:
 \begin{definition}
 If $\rho^\ast \gg \rho_0$, then the diffusion process is called global cascade.
 \end{definition}
 }

Recall that \cite{Watts2002} focuses on vanishingly small seeds (i.e., $\rho_0=0$).
The generalization to the case of $\rho_0>0$ would be useful when one investigates whether an activity/information that has already gained some popularity would further spread in a social network.
For example, given the fact that the number of active users of Slack exceeds 10 million, the question is: Will the population of Slack users increase further through a social network of friends/colleagues?
To formally answer this question in the cascade model, it is more natural to assume that the initial seed fraction $\rho_0$ is positive rather than 0. 
On the other hand, recall that the irreversibility of action is \emph{assumed} for initial seeds.
Practically, the irreversibility assumption for a finite fraction of seeds would be justified even when the seeds may possibly revert their action, because if some of the ``original'' seeds revert their action to be inactive, then we could consider the remaining active fraction as the ``core'' seed fraction $\rho_0$\footnote{For example, the number of ``original'' seeds correspond to the cumulative number of Slack users up to $t=0$, while the number of ``core'' seeds is given by the number of active users at $t=0$. 
One could also introduce a possibility that active nodes may be forced to become inactive, but it is beyond the scope of this paper and left for future research. \cite{kobayashi2015trend} and \cite{Ruan2015PRL} study fractional threshold models in which some fraction of nodes are not responsive to the neighbor's states.}.

 \subsection{Cascade conditions}
 
 \subsubsection{Generalized first-order cascade \addr{(GFC)} condition}
 The message-passing equation~\eqref{eq:recursion_q} can be viewed as a difference equation, whose stability is related to the condition under which global cascades can occur. Specifically, if the derivative of $G$ evaluated at the initial value $q_0=\rho_0$ is greater than 1, then there will be no stable fixed point such that $\rho_0=G(\rho_0)$. In this case, the fixed point $q^*$ will be apart from $\rho_0$, and therefore a global contagion happens, resulting in $q^*>\rho_0$. 
 \begin{theorem}
Let $F:\mathbb{Z}_{\geq 0}\times\mathbb{Z}_{>0} \to [0,1]$ satisfy Assumption~\ref{ass:F}.  If $\lim_{q\downarrow\rho_0} G^\prime(q)>1$, then there is no stable fixed point such that $\rho_0=G(\rho_0)$ for any $\rho_0\in[0,1)$.
\label{th:cascade_condition1st}
 \end{theorem}
 \begin{proof}
\addr{  See \ref{sec:proof_cascade_1st}.}
 \end{proof}
 
 We call the condition $\lim_{q\downarrow\rho_0} G^\prime(q)>1$ the \emph{generalized first-order cascade condition} since it is essentially a generalized version of the standard cascade condition proposed by \cite{Watts2002}, \cite{Gleeson2007}, and \cite{lelarge2012diffusion}:
 \begin{definition}
 For a given $\rho_0\in[0,1)$, the condition $\lim_{q\downarrow\rho_0} G^\prime(q)>1$ is called the generalized first-order cascade \addr{(GFC)} condition, which is rewritten as
 \begin{align}
      \sum_{k=2}^{\infty}\frac{k}{z}p_{k} \sum_{s=0}^{k-2}(k-1-s)\mathcal{B}_{s}^{k-1}(\rho_0)\left[F(s+1,k)-F(s,k)\right] >  1. 
    \label{eq:first-order}
 \end{align}
 A parameter space that satisfies this condition is called the first-order cascade region.
 \end{definition}
 See Eq.~\eqref{eq:Gprime_final} for a derivation of \eqref{eq:first-order}.
 Indeed, if $\rho_0=0$, the \addr{GFC} condition reduces to the conventional cascade condition~\citep{Watts2002,Gleeson2007,lelarge2012diffusion}:
 \begin{align}
     \sum_{k=2}^{\infty}\frac{k}{z}p_{k}(k-1)F(1,k) >  1. 
    \label{eq:first-order_standard}
 \end{align}
 \cite{Watts2002} obtains \eqref{eq:first-order_standard} using a generating-function approach, and \cite{Gleeson2007} derive it based on a message-passing method. They evaluate the derivative at $q=0$, assuming that $\rho_0$ is small enough, while condition \eqref{eq:first-order} allows us to evaluate the derivative at $q=\rho_0> 0$. 
 

 \subsubsection{Generalized extended cascade \addr{(GEC)} condition}
 
 Since condition~\eqref{eq:first-order} is based on the first-order derivative of $G(q)$ at $\rho_0$, the nonlinearity of $G$ near $q=\rho_0>0$ is ignored. This implies that the first-order condition may not always be able to detect a large-size cascade that could happen even if $G^\prime(\rho_0)\leq 1$ for $\rho_0>0$. 
 As suggested by \cite{Gleeson2007}, the accuracy of the condition would be improved if we exploit a second-order approximation.

 Let us rewrite $G(q)$ as $G(q) = \rho_0 + (1-\rho_0)S(q)$, where $S(q)$ denotes the nonlinear term. The second-order Taylor expansion of $S(q)$ about $\rho_0$ leads to $S(q) = S(\rho_0) + S^\prime(\rho_0)(q-\rho_0) + \frac{1}{2}S^{''}(\rho_0)(q-\rho_0)^2 = C_o + C_1(q-\rho_0) + C_2(q-\rho_0)^2,$ where $C_0\equiv S(\rho_0)$, $C_1\equiv S^\prime(\rho_0)$ and $C_2\equiv \frac{1}{2}S^{''}(\rho_0)$. 
 Note that the first derivative of $S$ on $q\in(0,1)$ is $S^\prime(q)= (1-\rho_0)^{-1}G^{\prime}(q)$ (see, Eq.~\ref{eq:Gprime_final}), and the second derivative is given by
 \begin{align}
     S^{''}(q) = &\sum_{k=3}^{\infty}\frac{k}{z}p_{k}\sum_{s=0}^{k-3}\binom{k-1}{s}(k-1-s)(k-2-s)q^s(1-q)^{k-3-s} \notag\\
     &\;\;\times \left[F(s+2,k)-2F(s+1,k)+F(s,k)\right]. 
     \label{eq:S_second}
 \end{align}
  See \ref{sec:second_derivative} for a derivation of \eqref{eq:S_second}.
 Up to the second-order approximation, $q=G(q)$ does not have a real root around $\rho_0\in(0,1)$ if \addr{its discriminant $D$ evaluated at $q=\rho_0$ is negative:} 
 \begin{align}
     \addr{D|_{q=\rho_0}\equiv\;} h_1^2 -4h_0h_2 < 0,
     \label{eq:cascade_condition2nd}
 \end{align}
 where $h_0\equiv\rho_0+(1-\rho_0)(C_0-C_1\rho_0+C_2\rho_0^2)$, $h_1\equiv (1-\rho_0)(C_1-2\rho_0C_2)-1$, and $h_2 \equiv (1-\rho_0)C_2$. By the continuity of $S^\prime$ and $S^{''}$ on $[0,1]$, this condition also holds for $\rho_0=0$ once we redefine as $C_1\equiv \lim_{q\downarrow 0}S^\prime (q)$ and $C_2\equiv (1/2)\lim_{q\downarrow 0}S^{''} (q)$.    
 Then, we can define the \emph{generalized extended cascade condition} as follows:
 \addr{
 \begin{definition}
 For a given $\rho_0\in[0,1)$, the generalized extended cascade \addr{(GEC)} condition is given by:
 \begin{align}
 \lim_{q\downarrow \rho_0}G^\prime (\rho_0)>1 \;\;\text{or}\;\; D|_{q=\rho_0}<0.
 \label{eq:GEC_definition}
 \end{align}
 A parameter space that satisfies \eqref{eq:GEC_definition} is called the extended cascade region.
 \end{definition}
 }
 \addr{The GEC condition is a generalized version of the extended cascade condition proposed by \cite{Gleeson2007}, which is recovered by setting $\rho_0=0$.
 We will show in section~\ref{sec:numerical_monoplex} that the \addr{GEC} condition predicts the simulated cascade region almost exactly.
 Note that in the GEC condition, both the first- and second-order approximations are taken into account to more accurately detect the parameter space in which a large-size cascade can occur. 
 Intuitively, the GFC condition (i.e., $\lim_{q\downarrow \rho_0}G^\prime (\rho_0)>1$) indicates whether the slope of $G(q)$ is larger than 1 at $q=\rho_0$ (Fig.~\ref{fig:bifurcation_Watts}). 
 In the GEC condition, on the other hand, a second-order approximation is also used to check whether there exists a real root of $G(q)-q=0$ around $\rho_0$ up to the second order (Eq.~\ref{eq:cascade_condition2nd}). 
 This improves the accuracy of the cascade region since it allows us to capture a situation where the GFC condition is not satisfied while a global cascade occurs, which could happen only if $D|_{q=\rho_0}<0$.
 This suggests that the GEC condition~\eqref{eq:GEC_definition} is necessary to ensure that $q^\ast$ does not exist around $\rho_0$ up to the second order.
 }

 \subsection{A mean-field approximation}\label{sec:naive}

 \begin{figure}[tb]
     \centering
     \includegraphics[width=12cm]{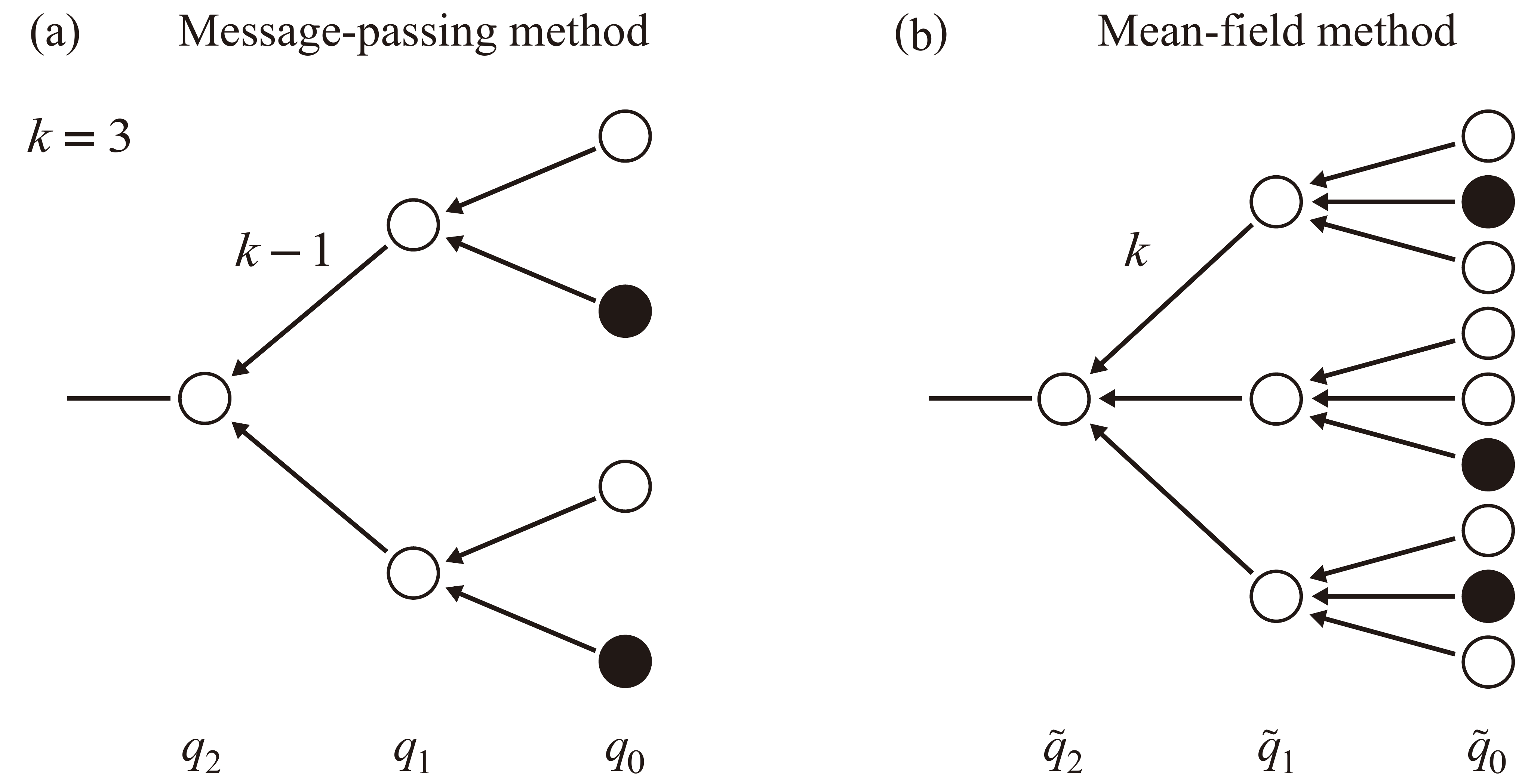}
     \caption{Schematic illustrations of the message-passing and mean-field methods in the case of a 3-regular graph (i.e., all nodes have degree $k=3$). }
     \label{fig:schematic_comparison}
 \end{figure}

As an alternative approach to the message-passing method, we also examine a naive mean-field approach that has been heavily used in the literature of network games~\citep{jackson2007diffusion,Jackson2008book,jackson-zenou2015games}. Recall that in the message-passing approach, we define $q_t$ as the probability that a neighbor randomly selected from the $k-1$ neighbors is active (Fig.~\ref{fig:schematic}). On the other hand, in the conventional mean-field approach, the counterpart of $q_t$, denoted by $\tilde{q}_t$, is defined as the probability of a neighbor randomly selected from \emph{all neighbors} being active\footnote{See, \cite{gleeson2018message} for further discussion.}.
A neighbor's activation probability $\tilde{q}_t$ is then computed using the average probability that the neighbor's neighbors are active, which is $\tilde{q}_{t-1}$. 
$\tilde{q}_t$ is thus given by the following recursion equation:
 \begin{align} 
 \tilde{q}_t = \rho_0 + (1-\rho_0)\sum_{k=1}^{\infty}\frac{k}{z}p_k \sum_{m=0}^{k} \mathcal{B}_{m}^{k}(\tilde{q}_{t-1})F\left(m,k \right),
 \label{eq:recursion_naive_q}
 \end{align}
 Note that the second summation runs from $m=0$ to $k$ as opposed to $k-1$, since all neighbors are equally taken into account in updating the activation probability.

To illuminate the essential difference between the mean-field and message-passing methods, we provide a schematic of how peer effects are computed in the case of a 3-regular graph in which all nodes have degree $3$ (Fig.~\ref{fig:schematic_comparison}).
\add{Fig.~\ref{fig:schematic_comparison}a illustrates that a node with degree~3 is affected only by two of its neighbors while sending a ``message'' to the remaining one.
On the other hand, in the mean-field method (Fig.~\ref{fig:schematic_comparison}b), each node receives messages from three neighbors while sending a message to ``another neighbor''. This means that each node essentially has one ``extra'' (or ``spurious'') edge, and the presence of such extra edges can become the source of inaccuracy. 
Therefore, to follow the flow of influence in a way that is consistent with the actual connectivity (i.e., $k=3$), one needs to remove the influence from one neighbor by taking a sum from $m=0$ to $k-1$ in the message-passing equation~\eqref{eq:recursion_q}. 
}

 
 Note that a common idea of the mean-field and message-passing methods is that they compute the average cascade size over all possible network structures that would be realized given a particular model (e.g., Erd\H{o}s-R\'{e}nyi, random regular graph, etc). 
 In other words, any realized network is deemed as a random sample from the network ensemble defined by the model (or more specifically, the degree distribution $p_k$).
 However, this does not mean that different networks are generated during the process of diffusion. 
 Once a network is realized (i.e., a sample is drawn) at $t=0$, the network structure is fixed throughout the diffusion process.
 By taking the average over $k$ with a prespecified degree distribution $p_k$ as in Eq.~\eqref{eq:rho_ss}, we can essentially calculate the ensemble average of cascade sizes over all possible network structures that would be realized under the model. 
 In the numerical experiments shown in the next section, we repeatedly simulate diffusion processes in each of which a network structure is randomly sampled at $t=0$.
 The average of simulated cascade sizes over different runs then corresponds to the theoretical ensemble average obtained by the mean-field/message-passing method.

 \subsection{Numerical experiments}\label{sec:numerical_monoplex}
 
 A stable fixed point $q^*=G(q^*)$ allows us to obtain the fraction of active players in the steady state, denoted by $\rho^*$, through Eq.~\eqref{eq:rho_ss}. However, it is not guaranteed that $\rho^*$ is identical to the ``true'' fraction of active players because in mean-field approximations, the network structure is captured by its degree distribution $p_k$ without incorporating detailed information about how players are connected.
  To examine how accurately the proposed method can explain the ``true'' fraction of active players, we run simulations of contagion on synthetic networks. Here, we employ \ER networks and calculate the average fraction of active players, denoted by $\hat{\rho}$, as follows:
 \begin{enumerate}
     \item For a given $z$ and $N$, generate an \ER network with the connecting probability $p=z/(N-1)$. Initially, all players are inactive.
     \item Select $\lfloor \rho_0 N\rfloor$ seed nodes at random and let them be active.
     \item \new{Update the state of all nodes, except for the seed nodes, simultaneously based on the threshold condition.}
     \item Repeat step 3 until convergence, where further updates do not change the states of players.  
     \item Repeat steps 1--4 and take the average of the final fractions of active players. The average is denoted by $\hat{\rho}$.
 \end{enumerate}
 It should be noted that the structure of an \ER network is characterized only by its degree distribution, so different networks having the same degree distribution are expected to yield the same simulation result on average. If network size $N$ is large enough, then the degree distribution of \ER networks having a common connecting probability $p=z/(N-1)$ will be Poissonian with mean $z$, indicating that the message-passing equation based on the Poissonian degree distribution could be applied to any \ER networks with average degree $z$. We set $N=10^4$ throughout the analysis.

 \subsubsection{Phase transitions}

 \begin{figure}[tb]
     \centering
     \includegraphics[width=15cm]{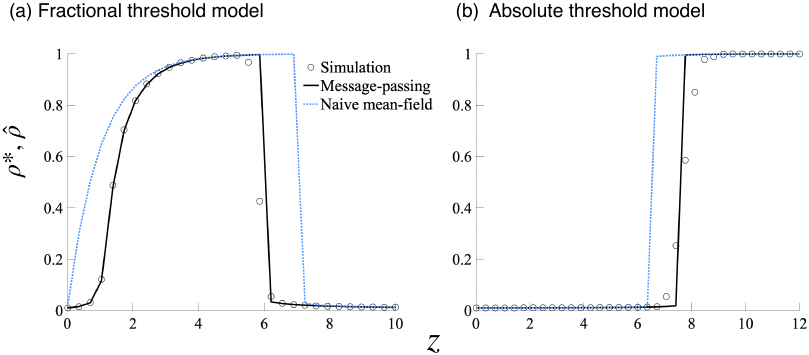}
     \caption{Steady-state fraction of active players in theory ($\rho^*$) and numerical experiments ($\hat\rho$). For (a) \new{the fractional threshold model} and (b) \new{absolute threshold model}, the thresholds are set at $\phi=0.2$ and $\theta=1.5$, respectively. $z$ denotes the average degree of \ER networks. The message-passing and naive mean-field approaches are denoted by black solid and blue dotted lines, respectively.
     Simulated fraction $\hat{\rho}$ (open circle) shows the average of 100 simulations. $N=10^4$ and $\rho_0=0.01$.}
     \label{fig:phase_transition_monoplex}
 \end{figure}
 
  Fig.~\ref{fig:phase_transition_monoplex}a shows the average of simulated fractions of active players (open circle) in the \new{fractional threshold model}. As is well known in the literature, there are two critical points with respect to the average degree $z$~\citep{Watts2002}. The first critical point is around $z=1$, below which global cascades cannot occur. This is because a network is not well connected if $z<1$, in which there are many isolated components. For a finite fraction of players to be active in a sufficiently large network, a vast majority of players needs to form a giant component~\citep{Gleeson2008}. The second critical point appears well above $z=1$, around $z=6$ in our case, above which global cascades will not happen. The reason is that if the network is dense enough, the influence from a single neighbor is diluted simply because the activation threshold is given as a share of active neighbors among $k$ neighbors.
 On the other hand, in the \addr{model of} absolute-threshold contagion, there is only one critical point (Fig.~\ref{fig:phase_transition_monoplex}b). This is because the activation threshold depends only on the total number of active neighbors $m$ independently of the number of neighbors $k$.

 Fig.~\ref{fig:phase_transition_monoplex} also shows that the message-passing method is much more accurate than the naive mean-field method in explaining the simulation results. In general, the naive mean-field method overestimates the cascade region because it fails to eliminate the repercussions of peer effects, as pointed out by \cite{gleeson2018message}. 
 \add{Such differences in accuracy are also observed for networks with scale-free degree distributions (see Fig.~\ref{fig:vs_z_SF}).}
 For this reason, we focus on the message-passing method in the following analysis.

 \subsubsection{Cascade region}\label{sec:cascade_region}

 \begin{figure}[p]
     \centering
     \includegraphics[width=15cm]{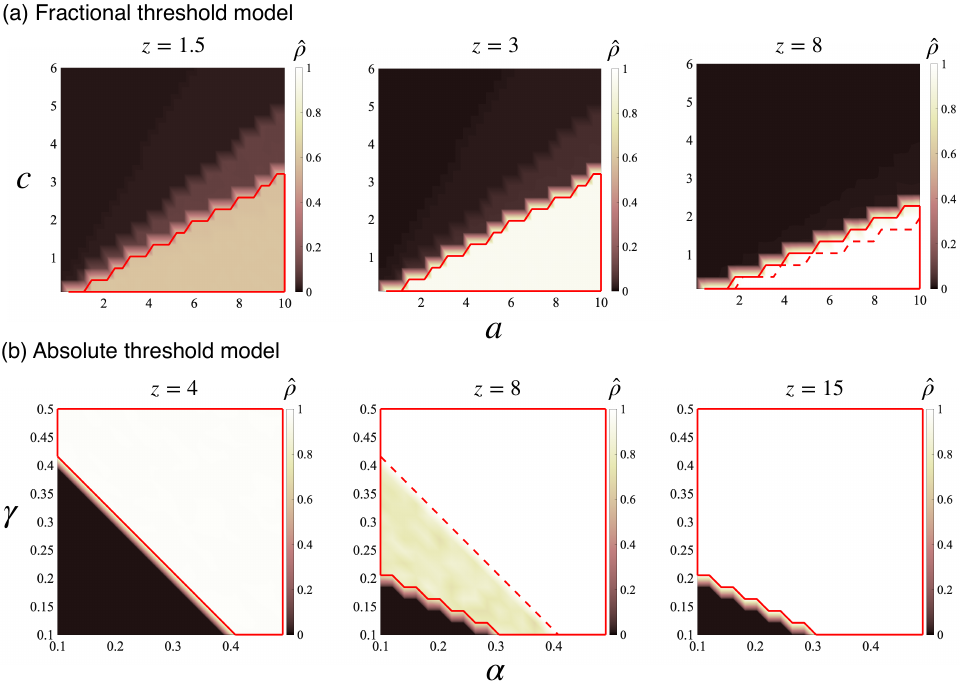}
     \caption{Simulated steady-state fraction of active players in (a) \new{the fractional threshold model} and (b) \new{absolute threshold model}. Red dashed and solid lines indicate the areas within which the \addr{GFC and GEC} conditions are satisfied, respectively (Red dashed is invisible when it is overlapped with solid). 
     For a given combination of parameters, we run 100 simulations to calculate the average of $\hat\rho$ (color bar). $N=10^4$ and $\rho_0=0.01$.}
     \label{fig:colormap_watts_grano}
 \end{figure}
 
 \begin{figure}[p]
     \centering
     \includegraphics[width=12cm]{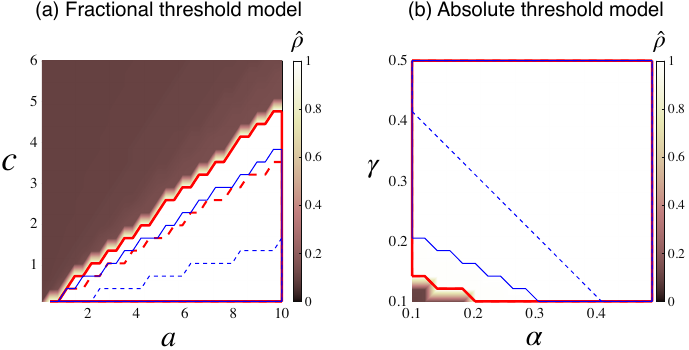}
     \caption{Comparison between generalized cascade conditions and the conventional ones. Red and blue lines denote the generalized and the standard cascade conditions, respectively. Solid and dashed respectively denote the extended and the first-order conditions. In panel (b), red dashed is fully overlapped with red solid. 
      $N=10^4$, \add{$z=8$}, and $\rho_0=0.1$.}
     \label{fig:rho10pt}
 \end{figure}

 Next, we examine how well the analytical cascade conditions predict the parameter space within which a global cascade can take place.
 Simulated fractions of active players in the steady states for given combinations of the payoff parameters are shown in Fig.~\ref{fig:colormap_watts_grano}.
The figure reveals that the \addr{GEC} condition (red solid) captures the simulated cascade region quite accurately in all the cases that we examined, while the \addr{GFC} condition (red dashed) may be less accurate (e.g., Fig.~\ref{fig:colormap_watts_grano}a, \emph{right}, and Fig.~\ref{fig:colormap_watts_grano}b, \emph{middle}). This suggests that the nonlinearity of the recursion equation $G$ would need to be exploited to predict the whole area of the true cascade region.

 On the other hand, when there are two distinct cascade regions in which the simulated values of $\hat\rho$ are apparently different, for instance, the middle panel of Fig.~\ref{fig:colormap_watts_grano}b, the \addr{GFC} condition correctly indicates the region in which the simulated fraction $\hat\rho$ is almost 1. In general, the \addr{GFC} condition is a more conservative criterion than the \addr{GEC} condition, but in many cases, the two conditions coincide and indicate exactly the same region as shown in Fig.~\ref{fig:colormap_watts_grano}.  
 
 Fig.~\ref{fig:rho10pt} illustrates how the generalized conditions can differ from the conventional cascade conditions proposed by \cite{Watts2002} and \cite{Gleeson2007}, among others. As we already pointed out, the conventional cascade conditions (blue lines) are obtained by differentiating the recursion equation at $q=0$. However, if $\rho_0$ is not sufficiently small ($\rho_0=0.1$ in Fig.~\ref{fig:rho10pt}), only the \addr{GEC} condition (red solid) would correctly predict the true cascade region, and the \addr{GFC} condition (red dashed) turns out to be as accurate as the standard extended condition (blue solid).

 \subsubsection{Dynamics of spreading process}

 \begin{figure}[tb]
     \centering
     \includegraphics[width=10cm]{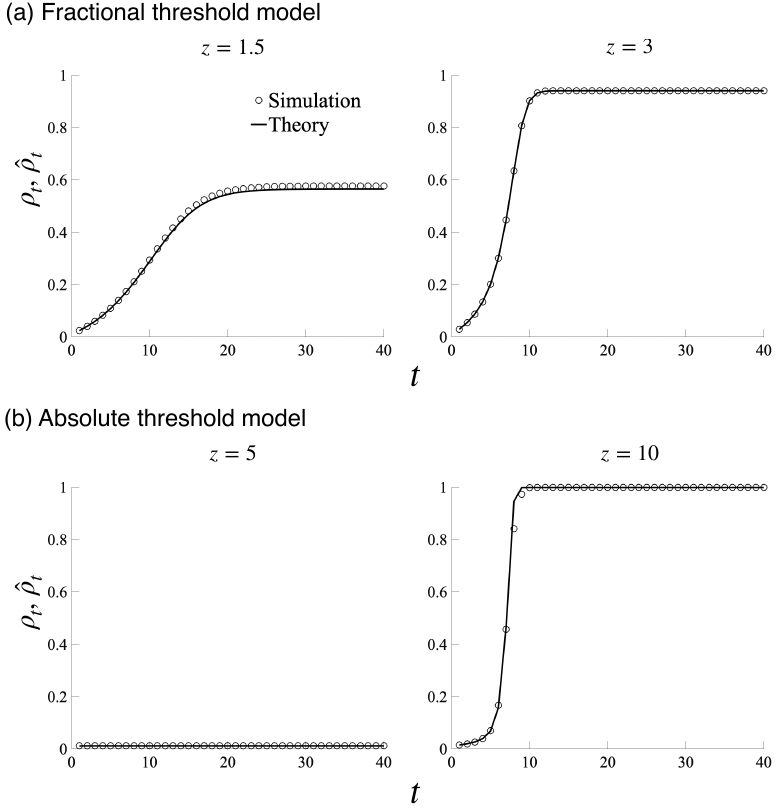}
     \caption{Dynamic paths of the fraction of active players in theory and simulation. See the caption of Fig.~\ref{fig:phase_transition_monoplex} for details.}
     \label{fig:path_monoplex}
 \end{figure}
 
 The message-passing equation $q_t=G(q_{t-1})$ is useful not only for predicting the steady-state fraction of active players, but also for describing the dynamics of contagion. 
 In fact, while iterating the recursion equation for $q_t$ is supposed to capture a spreading process on a hypothetical tree, the obtained value of $\rho_t$ from Eq.~\eqref{eq:rhot} well matches the simulated share of active players at $t$ in \ER networks (Fig.~\ref{fig:path_monoplex}), whose structures are not necessarily tree-like. This illustrates the correspondence between iterating recursion equation \eqref{eq:recursion_q}  and updating the states of players in numerical simulations. In a simulation of contagion, we simultaneously update the states of players at each time step, given the current states of their neighbors. The state-updating procedure at step $t$ corresponds to updating the recursion equation by substituting $q_{t-1}$ with $q_t$ in the RHS of Eq.~\eqref{eq:recursion_q}.

 \subsection{Social welfare}

In the current models, social welfare is measured by the sum of the realized payoffs/utilities. 
In the fractional threshold model, the social welfare $\mathcal{W}^{\rm frac}$ is given by
\begin{align}
    \mathcal{W}^{\rm frac} =
    \sum_{i\in\mathcal{N}^{\rm seed}} [-c(k_i-m_i)+am_i] + \sum_{i\notin\mathcal{N}^{\rm seed}} \max\{0,-c(k_i-m_i)+am_i\}, 
    \label{eq:welfare_frac_mono}
\end{align}
where $\mathcal{N}^{\rm seed}$ denotes the set of seed nodes.
Since seed nodes are always active independently of the states of neighbors, the sum of the seeds' payoffs are given by the first term of Eq.~\eqref{eq:welfare_frac_mono}. 
The second term comes from the fact that the payoffs of the players that are not selected as seeds are $0$ if they are inactive and $-c(k-m)+am$ if active. 
 
While the social welfare can be obtained by running numerical simulations, we can also calculate the average value of welfare based on the analytical methods. 
Using $q^\ast$ obtained by an analytical method, we can calculate the average (steady-state) social welfare $\widehat{\mathcal{W}}^{\rm frac}$ as follows:
\begin{align}
  \widehat{\mathcal{W}}^{\rm frac} &= \rho_0 N\sum_{k=0}^\infty p_k \sum_{m=0}^k\mathcal{B}_m^k(\qast) [-c(k-m)+am] \notag \\
   &\;\; + (1-\rho_0)N \sum_{k=0}^\infty p_k\sum_{m=0}^k\mathcal{B}_m^k(\qast) F(m,k)[-c(k-m)+am],
\label{eq:welfare_mono_theory}
\end{align}
where the first and second terms respectively correspond to those of Eq.~\eqref{eq:welfare_frac_mono}.

In the absolute threshold model, the social welfare is given by
\begin{align}
  {\mathcal{W}}^{\rm abs} = 
    \alpha\sum_i x_i -\frac{1}{2}\sum_i x_i^2 +\gamma\sum_i\sum_j\mathcal{A}_{ij}x_ix_j.
    \label{eq:welfare_abs}
\end{align}
Note that in the steady state, the theoretical average of the first term is $\alpha\rho^\ast N$ while that of the second term leads to $-\rho^\ast N/2$ since $x_i\in\{0,1\}$.
The third term represents ($2\gamma$ times) the number of edges between active nodes.
Since the total number of non-zero elements in the adjacency matrix is given by $Nz$, all we need to calculate the average of the third term is the average probability that a randomly chosen edge has active nodes on its both ends.
For this purpose, let $P(k,k^\prime)$ be the probability that a randomly chosen edge has a $k$-degree node and a $k^\prime$-degree node on its ends.
We also let $\rho_k^\ast$ denote the probability that a $k$-degree node is active in the steady state:
\begin{align}
    \rho_k^\ast = \rho_0 + (1-\rho_0)\sum_{m=0}^k\mathcal{B}_m^k(\qast)F(m,k).
\end{align}
Note that since we consider a class of random networks with prespecified degree distribution $p_k$, there is no degree-degree correlations~\citep{newman2018book2nd}.
The probability $P(k,k^\prime)$ is thus given by
$\frac{kp_k}{z}\frac{k^\prime p_{k^\prime}}{z}$.
It follows that the probability of a randomly chosen edge having $k$-degree and $k^\prime$-degree active nodes on its ends leads to $P(k,k^\prime)\rho_k^\ast\rho_{k^\prime}^\ast$.
Therefore, the theoretical average of social welfare in the absolute threshold model is given by
\begin{align}
     \widehat{\mathcal{W}}^{\rm abs} = \alpha\rho^\ast N -\frac{1}{2}\rho^\ast N +\gamma Nz\sum_{k=0}^\infty\sum_{k^\prime=0}^\infty P(k,k^\prime)\rho_k^\ast\rho_{k^\prime}^\ast.
\end{align}

 \begin{figure}[tb]
     \centering
     \includegraphics[width=11cm]{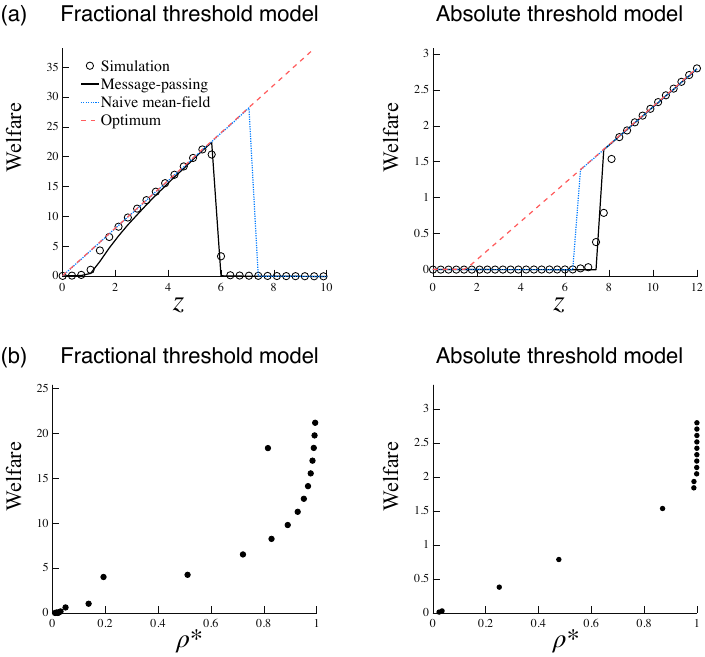}
     \caption{\add{Per capita welfare in the steady state. (a) Welfare against mean degree $z$. In each panel, simulated average welfare is compared with the theoretical averages obtained by the message-passing (black solid) and the mean-field (blue dotted) methods. It also shows the social optimum (red dashed). (b) Scatter plot of welfare against the fraction of active nodes $\rho^\ast$.   }}
     \label{fig:welfare_monoplex}
 \end{figure}
 
 Fig.~\ref{fig:welfare_monoplex}a shows the theoretical and simulated values of the per capita welfare (i.e., $\mathcal{W}/N$) against mean degree $z$.
 Clearly, the region in which a positive welfare is attained corresponds to the cascade region. 
 In fact, the social welfare is increasing in $\rho^\ast$ more than proportionally (Fig.~\ref{fig:welfare_monoplex}b).
 Fig.~\ref{fig:welfare_monoplex}a also shows the socially optimal welfare (denoted by $\overline{\mathcal{W}}$) that would be achieved if a social planner would optimally select the states of all nodes (red dashed).
 For the fractional threshold model, the social optimum is obviously attained when all nodes are active (i.e., $m_i=k_i$ for all $i$), where $\overline{\mathcal{W}}^{\rm frac} = azN$.
 In the absolute threshold model, on the other hand, the exact solution of $\{x_i\}_{i=1}^N\in \{0,1\}^N$ for the maximization of $\mathcal{W}^{\rm abs}$ (Eq.~\ref{eq:welfare_abs}) is difficult because one needs to examine  $2^N$ patterns of combinations\footnote{\add{This class of optimization problem is called \emph{quadratic unconstrained binary optimization} (QUBO) or \emph{unconstrained binary quadratic programming} (UBQP) and is known to be NP-hard~\citep{kochenberger2014qubo}.}}.  
 Instead, we consider a constrained optimum where all nodes are in the same state (i.e., $x_i=0$ or $1$ for all $i$).
 Then we have $\overline{\mathcal{W}}^{\rm abs} = \max\{0,(\alpha-1/2+\gamma z)N\}$.
 Clearly, the state of $x_i=1$ for all $i$ is indeed socially optimal when $z$ is large enough because the third term in Eq.~\eqref{eq:welfare_abs} becomes a dominant factor.
  
 The discrepancy between $\mathcal{W}$ and $\overline{\mathcal{W}}$ in Fig.~\ref{fig:welfare_monoplex} indicates that depending on the density (i.e., the mean degree) of the network, \addr{the attained Nash equilibrium is not necessarily optimal in that the social welfare may not be maximized}.
 In the fractional threshold model (resp. the absolute threshold model), the fraction of active nodes is virtually zero for $z>6$ (resp. $z<7$), in which case the players would be better off if they could collectively coordinate on being active.
 This illustrates the fact that a coordination failure occurs outside the cascade region, where there is a positive externality that being active will benefit its neighbors by making it desirable for them to be active.
 It is also worth mentioning that the extent to which the network density affects social welfare depends on which diffusion process is considered.
 If the diffusion process is ruled by a fractional threshold, then increased density will have a negative impact on social welfare near a critical point (i.e., $z\approx 6$ in Fig.~\ref{fig:welfare_monoplex}a, \emph{left}).
 On the other hand, social optimum is attained only when $z$ is sufficiently high in the absolute threshold model (Fig.~\ref{fig:welfare_monoplex}a, \emph{right}).

 \section{Games on multiplex networks}

 \begin{figure}[tb]
     \centering
     \includegraphics[width=12cm]{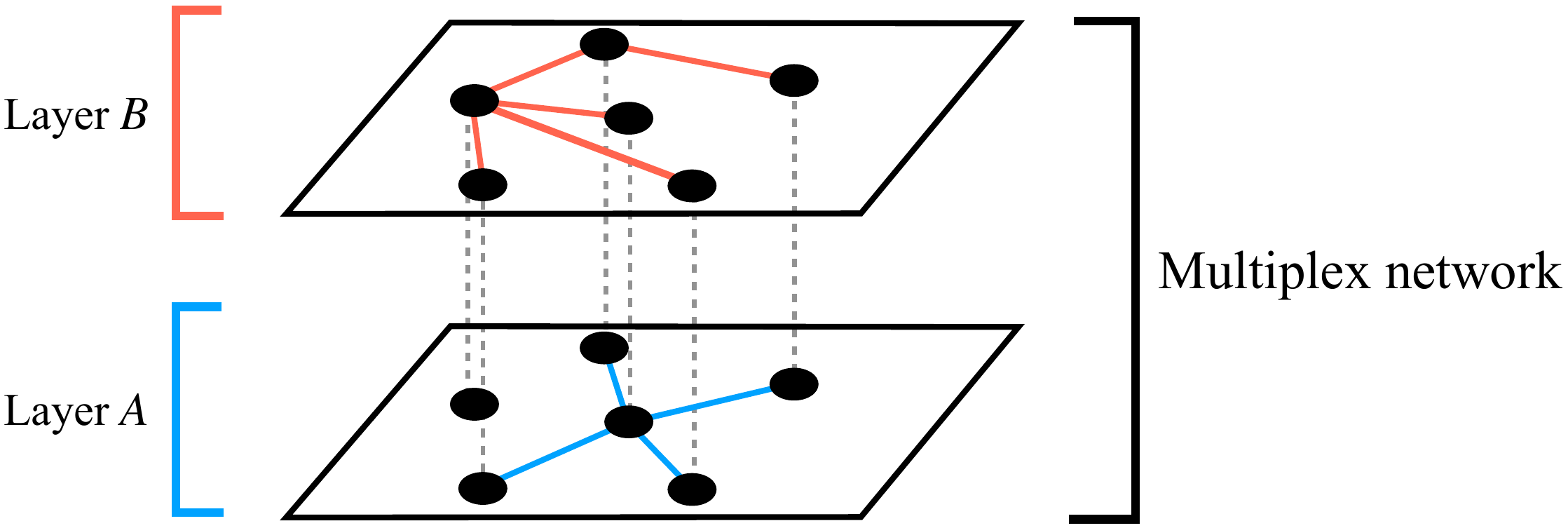}
     \caption{Schematic of a multiplex network. Different colors of edges represent different types of social ties. Two nodes connected by a dashed line denote a player.}
     \label{fig:schematic_multiplex}
 \end{figure}
 
 In reality, \addr{individuals} are connected to each other in a wide variety of social contexts. These include online spaces such as Twitter and Facebook, as well as physical spaces such as schools and work places.
 In network science, such situations are modeled as \emph{multiplex networks}, in which each layer represents a single network formed in a particular social context~\citep{bianconi2018book}\footnote{If a common set of nodes forms multiple networks, the set of networks (or layers) is called a multiplex network. If the sets of nodes in different layers are not \addr{necessarily} common, the networks are collectively called \emph{multilayer networks}.}.

 In this section, we extend the previous analysis to multiplex networks with two layers. As shown in the schematic (Fig.~\ref{fig:schematic_multiplex}), every player belongs to both the layers, and edges in different layers capture different sorts of social ties (e.g., friendships in school and follower-followee relationships in Twitter). We assume that a multiplex network consists of two independent \add{(sparse)} \ER networks.  This greatly simplifies the analysis not only because it guarantees that each network has a locally tree-like structure, but also because the presence of overlapped edges between different layers can be ignored if the networks are sufficiently large~\citep{Bianconi2013,bianconi2018book}. 
 
 \add{It should be noted that in some special cases, the dynamics of diffusion can be analyzed with an aggregated (one-layer) network since all nodes are common between the two layers. 
 In our model, if all the payoff parameters are identical across different layers, then the dynamics on a multiplex network can be described as that on the aggregated monoplex network created by combining the two layers. 
 In this section, we consider a general case in which the payoff parameters can vary across layers, taking into account different social contexts.
 }

 \subsection{\new{Fractional threshold model}}
 
 In a two-layer multiplex network, each player belongs to layer $A$ and layer $B$ and selects a strategy in a way that maximizes the total net payoff. The payoff matrix of a coordination game in a layer is given in the same way as that for a monoplex network (Table~\ref{tab:one-good game}). 
 The payoff of being active in layer $\ell\in\{A,B\}$ is thus given by
 \begin{align}
     -c_\ell(k_{\ell}-m_\ell) + a_\ell m_\ell,
 \end{align}where parameters $a_\ell$ and $c_\ell$ are respectively the payoff of coordination and the cost of failing to coordinate in layer $\ell$. $m_\ell$ denotes the number of active neighbors in layer $\ell$. 
 Since the payoff of being inactive is $0$, a player will decide to be active if the total net payoff of doing so is positive:
 \begin{align}
       -\left[c_A(k_{A}-m_A )+c_B(k_{B}-m_B )\right] +  a_A m_A + a_B m_B  > 0.
       \label{eq:response_condition_watts}
 \end{align}
 
  The response function $\mathcal{F}$ for \new{the fractional threshold model} on multiplex networks is then defined as
 \begin{align}
     \mathcal{F}(m_A,m_B,\vect{k})  \equiv  \mathcal{F}^{\rm frac}(m_A,m_B,\vect{k}) &=
     \begin{cases}
      1 \;\; \text{ if Eq.~\eqref{eq:response_condition_watts} is satisfied}, \\
      0 \;\; \text{ otherwise},
     \end{cases}
     \label{eq:F_watts_multi}
 \end{align}
where $\vect{k}=(k_A,k_B)$.
 \addr{Note} that this response function satisfies a multivariate version of Assumption~\ref{ass:F}, namely $\mathcal{F}(0,0,\vect{k})=0$, and $\mathcal{F}$ is \addr{increasing} in both $m_A$ and $m_B$.

\subsection{\new{Absolute threshold model}}

The utility function of player $i$ is given as
\begin{align}
    u_i(x_i;{\bf{x}}_{-i}) = \alpha x_i - \frac{1}{2}x_i^2 +\gamma_A\sum_{j\ne i} \mathcal{A}_{ij}^{(A)} x_ix_j 
    +\gamma_B\sum_{j\ne i} \mathcal{A}_{ij}^{(B)} x_ix_j, 
\end{align}
 where $\mathcal{A}_{ij}^{(\ell)}$ denotes the $(i,j)$th element of the adjacency matrix of layer $\ell$, and  $\gamma_\ell$ denotes the  weight on the influence from neighbors in layer $\ell$. If $\gamma_A \neq \gamma_B$, the impact of peer effects is different between the two layers (e.g., school friends may be more influential than Twitter followers).
 
 In binary-action games, player $i$ will become active if $u_i(0;{\bf{x}}_{-i})<u_i(1;{\bf{x}}_{-i})$, that is
 \begin{align}
     \alpha- \frac{1}{2} + \gamma_{A}m_{A,i}+\gamma_{B}m_{B,i} >0,
     \label{eq:condition_response_Grano}
 \end{align}
 where $m_{\ell,i}=\sum_{j\ne i} \mathcal{A}_{ij}^{(\ell)} x_j$ denotes the number of active neighbors in layer $\ell$. The response function is thus given by
 \begin{align}
     \mathcal{F}(m_A,m_B,\vect{k})  \equiv \mathcal{F}^{\rm abs}(m_A,m_B) &=
     \begin{cases}
      1 \;\; \text{ if }\;\; \alpha- \frac{1}{2} + \gamma_{A}m_{A}+\gamma_{B}m_{B} >0, \\ 
      0 \;\; \text{ otherwise}.
     \end{cases}
     \label{eq:F_Grano_multi}
 \end{align}
  It is clear that if $\alpha <1/2$, the response function satisfies $\mathcal{F}(0,0,\vect{k})=0$ and is \addr{increasing} in both $m_A$ and $m_B$.

\subsection{Contagion dynamics on multiplex networks}

Let $q_t^{\ell}$ denote the probability of a randomly chosen neighbor in layer $\ell\in\{A,B\}$ being active. The recursion equations for $q_t^A$ and $q_t^B$ are given by~\citep{Yagan2012,Brummitt2012_PRER}:
\begin{align}
    q_t^{A} &= \rho_0 + (1-\rho_0)\sum_{k_B=0}^\infty p_{k_B}\sum_{m_B=0}^{k_{B}}\mathcal{B}_{m_B}^{k_B}\left(q_{t-1}^{B}\right)\sum_{k_A=1}^\infty\frac{k_A p_{k_A}}{z_A}\sum_{m_A=0}^{k_A-1}\mathcal{B}_{m_A}^{k_A-1}\left(q_{t-1}^{A}\right)\mathcal{F}(m_A,m_B,\vect{k}), \notag \\
    &\equiv g^{(A)}(q_{t-1}^A,q_{t-1}^B), \label{eq:qtA}\\
    q_t^{B} &= \rho_0 + (1-\rho_0)\sum_{k_A=0}^\infty p_{k_A}\sum_{m_A=0}^{k_{A}}\mathcal{B}_{m_A}^{k_A}\left(q_{t-1}^{A}\right)\sum_{k_B=1}^\infty\frac{k_B p_{k_B}}{z_B}\sum_{m_B=0}^{k_B-1}\mathcal{B}_{m_B}^{k_B-1}\left(q_{t-1}^{B}\right)\mathcal{F}(m_A,m_B,\vect{k}), \notag \\
    &\equiv g^{(B)}(q_{t-1}^A,q_{t-1}^B).\label{eq:qtB}
\end{align}
Let us express the vector of functions $g^{(A)}$ and $g^{(B)}$ as: 
\begin{align}
\left(
\begin{array}{c}
q_t^A \\ q_t^B
\end{array}
\right)
 = 
 \left[
 \begin{array}{c}
 g^{(A)}(q_{t-1}^A, q_{t-1}^B) \\
 g^{(B)}(q_{t-1}^A,q_{t-1}^B)
\end{array}
\right].
\label{eq:multiplex_recursion}
\end{align}
The following proposition describes the boundary conditions of $g^{(\ell)}$.
\begin{proposition}
Let $\mathcal{F}:\mathbb{Z}_{\geq 0}^4\to \{0,1\}$ satisfy $\mathcal{F}(0,0,\vect{k})=0$. Then $g^{(\ell)}(0,0)=\rho_0$ and $g^{(\ell)}(1,1)\leq 1$ for $\ell=A,B$.
\label{prop:bounded_multi_g}
\end{proposition}
\begin{proof}
 $g^{(A)}(0,0) = \rho_0+(1-\rho_0)\sum_{k_B=0}^\infty p_{k_B}\sum_{k_A=1}^\infty\frac{k_A p_{k_A}}{z_{A}}\mathcal{F}(0,0,\vect{k})=\rho_0,$ and $g^{(A)}(1,1) = \rho_0+(1-\rho_0)\sum_{k_B=0}^\infty p_{k_B}\sum_{k_A=1}^\infty\frac{k_A p_{k_A}}{z_{A}}\mathcal{F}(k_A-1,k_B,\vect{k})\leq 1.$ Because of the symmetry, we also have $g^{(B)}(0,0) =0$ and $g^{(B)}(1,1) \leq 1$.
\end{proof}

For $q^A,q^B\in(0,1)$, the Jacobian matrix of vector $\left[g^{(A)}(q^A,q^B),g^{(B)}(q^A,q^B)\right]^\top$ is given by
\begin{align}
    \J(q^A,q^B) \equiv \left[
    \begin{array}{cc}
        \J_{11}(q^A,q^B) & \J_{12}(q^A,q^B) \\
        \J_{21}(q^A,q^B) & \J_{22}(q^A,q^B)
    \end{array}
    \right] = \left[
    \begin{array}{cc}
        \frac{\partial g^{(A)}}{\partial q^A} & \frac{\partial g^{(A)}}{\partial q^B} \\
        \frac{\partial g^{(B)}}{\partial q^A} & \frac{\partial g^{(B)}}{\partial q^B}
    \end{array}
    \right].
    \label{eq:Jacobian}
\end{align}
where each element of the Jacobian matrix is given in \ref{sec:jacobian_appendix}. 
From the monotonicity of $\mathcal{F}$, we have the next proposition:
\begin{proposition}
If $\mathcal{F}:\mathbb{Z}_{\geq 0}^4\to \{0,1\}$ is \addr{increasing} in the first two arguments, then $g^{(A)}$ and $g^{(B)}$ are \addr{increasing}.
\label{prop:monotone_multi_g}
\end{proposition}
\begin{proof}
 See, \ref{sec:jacobian_appendix}.
\end{proof}

From Propositions~\ref{prop:bounded_multi_g} and \ref{prop:monotone_multi_g},  $g^{(A)}$ and $g^{(B)}$ are \addr{increasing} and bounded on $[\rho_0,1]$. Therefore, combined with their continuity, recursion equations \eqref{eq:qtA} and \eqref{eq:qtB} converge to fixed points $(q^{A*},q^{B*})\in [\rho_0,1]^2$ if we iterate them from the initial condition $(q_0^A,q_0^B)=(\rho_0,\rho_0)$.

The cascade dynamics on multiplex networks are described by the vector-based recursion equation \eqref{eq:multiplex_recursion}, so the first-order cascade condition is that the largest eigenvalue of the Jacobian $\J$ evaluated at $q^A=q^B=\rho_0$ exceeds 1:
\begin{align}
    \lambda_{\rm max}\left[\J(\rho_0,\rho_0)\right] = \frac{\J_{11,\rho_0}+\J_{22,\rho_0}+\sqrt{\left(\J_{11,\rho_0}-\J_{22,\rho_0}\right)^2+4\J_{12,\rho_0}\J_{21,\rho_0}}}{2}>1,
    \label{eq:eig_condition}
\end{align}
where $\J_{uv,\rho_0}\equiv \J_{uv}(\rho_0,\rho_0)$. From \ref{sec:jacobian_appendix}, $\J_{11,\rho_0}$ and $\J_{12,\rho_0}$ respectively lead to
\begin{align}
    \J_{11,\rho_0} &= \sum_{k_B=0}^\infty p_{k_B}\sum_{m_B=0}^{k_{B}}\mathcal{B}_{m_B}^{k_B}\left(\rho_0\right)\sum_{k_A=2}^\infty\frac{k_A p_{k_A}}{z_A} \sum_{m_A=0}^{k_A-2}(k_A-1-m_A)\mathcal{B}_{m_A}^{k_A-1}(\rho_0) \notag \\
    &\hspace{.7cm} \times 
    \left[ \mathcal{F}(m_A+1,m_B,\vect{k}) - \mathcal{F}(m_A,m_B,\vect{k})\right],\\
     \J_{12,\rho_0} &= \sum_{k_A=1}^\infty \frac{k_A p_{k_A}}{z_A}\sum_{m_A=0}^{k_{A}-1}\mathcal{B}_{m_A}^{k_A-1}\left(\rho_0\right)\sum_{k_B=1}^\infty p_{k_B}  \sum_{m_B=0}^{k_B-1}(k_B-m_B) \mathcal{B}_{m_B}^{k_B}(\rho_0)\notag \\
    &\hspace{.7cm}\times \left[ \mathcal{F}(m_A,m_B+1,\vect{k}) - \mathcal{F}(m_A,m_B,\vect{k})\right].
\end{align}
$\J_{21,\rho_0}$ and $\J_{22,\rho_0}$ are obtained analogously.
Due to the continuity of functions $g^{(A)}(q^A,q^B)$ and $g^{(B)}(q^A,q^B)$ on $[0,1]^2$,  condition~\eqref{eq:eig_condition} based on the Jacobian can be applied for any $\rho_0\in [0,1)$. Now we can define a cascade region as follows:
\begin{definition}
For a given seed fraction $\rho_0\in[0,1)$, the cascade region for a multiplex network is a parameter space within which $\lambda_{\rm max}[\J(\rho_0,\rho_0)]>1$ is satisfied.
\end{definition}

As an extension of the first-order cascade condition for monoplex networks, some studies developed an eigenvalue-based cascade condition for multiplex networks~\citep{Yagan2012,Brummitt2012_PRER}. In these studies, however, the Jacobian matrix is obtained under the assumption that $\rho_0=0$, where the derivative is evaluated at $q^A = q^B=0$. Condition~\eqref{eq:eig_condition} is thus a generalized cascade condition for multiplex networks, in which imposing $\rho_0=0$ will recover the conventional one.

\subsection{Numerical experiments}

 \addr{In this section, we} examine how well the analytical cascade condition~\eqref{eq:eig_condition} predicts the simulated cascade region in two-layer multiplex networks. To focus on the effect of multiplicity, we consider a multiplex network that consists of two independent \ER networks having the same degree distribution $p_k$. Note that if the parameters in the response function are identical between the two layers (i.e, $a_\ell$ and $c_\ell$ for the fractional threshold model, and $\gamma_\ell$ for the absolute threshold model), then a two-layer multiplex model reduces to a monoplex model in which the average degree is given by the sum of the two average degrees of the original two layers.
 Therefore, we allow these parameters to vary across different layers in the following manner:
 \begin{align}
     a_A = (1-\delta)a, \;\; a_B = (1+\delta)a,
 \end{align}
 for the fractional threshold model, and 
 \begin{align}
     \gamma_A = (1-\delta)\gamma, \;\; \gamma_B = (1+\delta)\gamma,
 \end{align}
 for the absolute threshold model, where $\delta\in[0,1]$ represents the extent of inter-layer heterogeneity in the peer effects. This specification ensures that the average of the layer-specific parameters is equal to the original value used in the monoplex model (e.g., $(a_A+a_B)/2=a$). 
 The other parameters are assumed to be the same across layers.


 We find that inter-layer heterogeneity in the threshold conditions expands the first-order cascade region in both classes of contagion (Fig.~\ref{fig:colormap_multiplex}). This is because the introduction of heterogeneity necessarily relaxes the threshold condition in one layer, while it is tightened in the other layer. This will let the former layer ``lead'' the spreading process, \addr{and the other layer ``follows''}. 
 On the other hand, as in the case of monoplex networks, the first-order cascade condition does not always cover the whole area of the simulated cascade region. If the \addr{average} size of \addr{simulated cascades} $\hat{\rho}$ is less than one, then the first-order condition based on the largest eigenvalue of the Jacobian may fail to indicate the correct cascade region  (Fig.~\ref{fig:colormap_multiplex}a, \emph{right}, and b, \emph{left}). In this sense, the first-order condition should be viewed as a conservative criterion, as we pointed out in the monoplex model.

 \section{Conclusion and discussion}
 
 In network science, the collective dynamics of cascading behavior on complex networks have been studied using message-passing approximations. In many of these studies, the accuracy of the proposed approximation methods is rigorously examined by comparing the theoretical predictions with the simulated results. However, cascade models often assume that individuals follow an exogenously given threshold rule \addr{without microfoundations}. In the literature on coordination games on networks, on the other hand, a similar mean-field method is occasionally used, but little quantitative validation has been done. In this paper, we provided a unifying framework that connects the two strands of literature by providing microfoundations for two classes of threshold rules. 
 \add{Our main contribution is threefold: i) we provide formal analytical proofs for the existence of a Nash equilibrium and the convergence of the iteration algorithm based on the widely used message-passing equation, ii) the cascade conditions are generalized to include more practically relevant cases in which the seed fraction is positive, and iii) we present a welfare analysis in which the social welfare is analytically calculated.}
 
 Some issues remain to be addressed in future research. First, while we focused on \addr{random networks} to apply a message-passing method, real social networks are not necessarily random. Actual social networks often exhibit degree correlations and community structures~\citep{Barabasi2016book,newman2018book2nd}. It would be worthwhile to study the impact of these realistic factors on the likelihood of contagion. In fact, some studies have reported that (modified) message-passing methods work well for networks that are not necessarily locally tree-like~\citep{ikeda2010cactus,melnik2011unreasonable}. Second, the current work does not consider a case of incomplete information. The preferences of other players may be unknown in reality, and in such cases, players would learn the states of the neighbors from their actions~\citep{sadler2020diffusion}. Third, endogeneity of network formation is absent in our model. In many social contexts, individuals can decide when and with whom to talk. This would require a dynamic network model in which players optimally select (temporal) neighbors in a way that maximizes their utilities~\citep{jacksonJEL2017}. We hope that our work will stimulate further research on the dynamics of contagion in these realistic environments.

\bigskip

\appendix
\renewcommand{\thesection}{Appendix \Alph{section}}
\renewcommand{\theequation}{\Alph{section}.\arabic{equation}}
\setcounter{equation}{0}

\section{Proof of Proposition~\ref{prop:G_monotonicity}}
\label{sec:proof_G_mononone}

The derivative of $G(q)$ for $q\in (0,1)$ leads to
\begin{align}
    G^\prime (q) &= (1-\rho_0)\sum_{k=2}^{\infty}\frac{k}{z}p_{k}\sum_{m=0}^{k-1}\binom{k-1}{m}\left[ mq^{m-1}(1-q)^{k-1-m}\right. \notag \\
    & \;\;\;\left. -(k-1-m)q^m(1-q)^{k-2-m})\right]F(m,k), \\
    &= (1-\rho_0)\sum_{k=2}^\infty \frac{k}{z}p_{k}\Gamma(q;k),
\end{align}
where we define function $\Gamma (q;k)$ for $k=1,2,\ldots,$ as 
\begin{align}
\Gamma(q;k)\equiv \sum_{m=0}^{k-1}\binom{k-1}{m}q^{m-1}(1-q)^{k-2-m} \left[ m(1-q)-(k-1-m)q\right]F(m,k).
\end{align}
Note that $\Gamma(q;1)=0$. To prove that $G(q)$ is \addr{increasing} in $q\in(0,1)$, it is sufficient to show that $\Gamma(q;k)$ is non-negative on $(0,1)$ for $k = 2,3,\ldots$.

To prove $\Gamma(q;k)\geq 0$, we first expand $\Gamma$ as
\begin{align}
    \Gamma(q;k) =& \;\;\;\; \binom{k-1}{0}q^{-1}(1-q)^{k-2}\left[0 - (k-1)q\right]F(0,k) \notag \\
     & + \binom{k-1}{1}q^{0}(1-q)^{k-3}\left[1(1-q) - (k-2)q\right]F(1,k) \notag \\
     & + \binom{k-1}{2}q^{1}(1-q)^{k-4}\left[2(1-q) - (k-3)q\right]F(2,k) \notag \\
     & + \cdots \notag \\
     & + \binom{k-1}{k-1}q^{k-2}(1-q)^{-1}\left[(k-1)(1-q) - 0\right]F(k-1,k), \notag\\
     &\notag \\
     =& \;\;\;\;\frac{(k-1)!}{0!(k-2)!}q^{0}(1-q)^{k-2}\left[F(1,k)-F(0,k)\right] \notag\\
    & +\frac{(k-1)!}{1!(k-3)!}q^{1}(1-q)^{k-3}\left[F(2,k)-F(1,k)\right] \notag \\
   & + \frac{(k-1)!}{2!(k-4)!}q^{2}(1-q)^{k-4}\left[F(3,k)-F(2,k)\right] \notag \\
   & + \cdots \notag \\
   & + \frac{(k-1)!}{(k-2)!0!}q^{k-2}(1-q)^{0}\left[F(k-1,k)-F(k-2,k)\right]. 
    \label{eq:Gamma_arrange}
\end{align}
Since $F$ is \addr{increasing} in its first argument (Assumption~\ref{ass:F}), it is clear that $F(m+1,k)-F(m,k)\geq 0$ for $m=0,\ldots k-2$ and therefore $\Gamma(q;k)\geq 0$ for $k=2,3,\ldots$. It follows that for any $q\in (0,1)$, we have
\begin{align}
    G^\prime(q) = (1-\rho_0)\sum_{k=2}^{\infty}\frac{k}{z}p_{k} \sum_{s=0}^{k-2}\binom{k-1}{s}(k-1-s)q^{s}(1-q)^{k-2-s}\left[F(s+1,k)-F(s,k)\right]\geq 0. 
    \label{eq:Gprime_final}
\end{align}
From the continuity of $G(q)$, it is true that $\lim_{q\downarrow 0}G(q) = G(0)$ and $\lim_{q\uparrow 1}G(q) = G(1)$. This proves that $G(q)$ is \addr{increasing} in $q\in[0,1]$. \qed
\bigskip

\addr{
\section{Proof of Theorem~\ref{th:fixed_point}}\label{sec:proof_fixed_point}
 We apply Kleene's fixed point theorem to $G$ on $[0,1]$.
 Let $\{q_t\}_{t\in\mathbb{Z}_{\geq 0}}$ be an increasing sequence in $[0,1]$ such that $q_t\leq q_{t+1}$.
 Since $0\leq q_t\leq 1$ for all $t\in\mathbb{Z}_{\geq 0}$, we have $\sup_{t\in\mathbb{Z}_{\geq 0}} q_t$ in $[0,1]$.
 Therefore, $[0,1]$ is $\omega$-complete.
 From the continuity and monotonicity of $G$ on $[0,1]$ (Proposition~\ref{prop:G_monotonicity}), we have $G\left({\rm sup}_{t\in\mathbb{Z}_{\geq 0}}q_t\right) 
 = {\rm sup}_{t\in\mathbb{Z}_{\geq 0}}G(q_{t})$ for any increasing sequence $\{q_t\}_{t\in\mathbb{Z}_{\geq 0}}$ in $[0,1]$. Thus, $G$ is $\omega$-continuous.

 It remains to show that $q_0 \leq G(q_0) $ for any initial value $q_0\in[0,\rho_0]$.
 Let $q_0$ satisfy $0\leq q_0\leq\rho_0$. 
 From the monotonicity of $G$, we have $G(0)\leq G(q_0)\leq G(\rho_0)$, which indicates that $q_0\leq\rho_0\leq G(q_0)$ because we have $G(0)=\rho_0$ (Proposition~\ref{prop:G_01}). 
 This proves that $q_0\leq G(q_0)$ for any $q_0\in[0,\rho_0]$.
 It follows from Kleene's fixed point theorem that the sequence $\{q_t\}_{t\in\mathbb{Z}_{\geq 0}}$ generate by $G$ for a given initial value $q_0\in[0,\rho_0]$ converges to the least fixed point $q^\ast =G(q^\ast)$, where $q^\ast=\sup_{t\in\mathbb{Z}_{\geq 0}}G^t(q_0)\in[0,1]$.
\qed
}

\addr{
\section{Proof of Theorem~\ref{th:cascade_condition1st}}
\label{sec:proof_cascade_1st}
First, consider the case of $\rho_0=0$. From Proposition~\ref{prop:G_01}, we have $G(0)=\rho_0=0$, so $0$ is a fixed point. If $\lim_{q\downarrow 0} G^\prime(q)>1$, then $\lim_{q\downarrow 0} G^\prime(q) = \lim_{q\downarrow 0}\lim_{\epsilon\downarrow 0}\frac{G(q+\epsilon)-G(q)}{\epsilon} =\lim_{\epsilon\downarrow 0}\frac{G(\epsilon)-G(0)}{\epsilon} = \lim_{\epsilon\downarrow 0}\frac{G(\epsilon)}{\epsilon}>1$. Thus, we have $\epsilon < G(\epsilon)$ in the neighborhood of the fixed point $\rho_0=0$. From the monotonicity of $G$, it follows that $\epsilon < G(\epsilon) < G(G(\epsilon))<\cdots$. This proves that $\rho_0=0$ is an unstable fixed point at which a small perturbation will cause a transition to a larger fixed point. 
  
  If $\rho_0>0$, we would have a fixed point $\rho_0 = G(\rho_0)$ if and only if $G(\rho_0)-G(0)=0$ since $G(0)=\rho_0$.  To satisfy $G(\rho_0)-G(0)=0$, $G$ must not be increasing on $[0,\rho_0]$ due to its monotonicity. Thus, $G(q+\epsilon)-G(q)=0$ must hold for any $q\in[0,\rho_0)$ and $\epsilon\in (0,\rho_0-q)$ for $\rho_0(>0)$ to be a fixed point. Then, we have $\lim_{q\uparrow\rho_0}G^\prime (q) = \lim_{q\uparrow\rho_0}\lim_{\epsilon\downarrow 0}\frac{G(q+\epsilon)-G(q)}{\epsilon}=\lim_{q\uparrow\rho_0}\lim_{\epsilon\downarrow 0}\frac{0}{\epsilon}=0.$
  It follows from the continuity of $G^\prime(q)$ on $(0,1)$ (see, Eq.~\ref{eq:Gprime_final}) that $\lim_{q\downarrow\rho_0} G^\prime(q) =\lim_{q\uparrow\rho_0} G^\prime(q)= 0$. This proves that $\rho_0>0$ is not a fixed point if $\lim_{q\downarrow\rho_0} G^\prime(q)>1$. \qed
}

\section{Second derivative $S^{''}(q)$ in Eq.~\eqref{eq:S_second}}\label{sec:second_derivative}

From Eq.~\eqref{eq:Gprime_final}, the second derivative of $G(q)$ on $(0,1)$ leads to
\begin{align}
    G^{''} (q) &= (1-\rho_0)\sum_{k=2}^{\infty}\frac{k}{z}p_{k}\sum_{m=0}^{k-2}\binom{k-1}{m}(k-1-m)\left[ mq^{m-1}(1-q)^{k-2-m}\right. \notag \\
    & \;\;\;\left. -(k-2-m)q^m(1-q)^{k-3-m})\right]\left[F(m+1,k)-F(m)\right], \notag \\
    &= (1-\rho_0)\sum_{k=2}^\infty \frac{k}{z}p_{k}\widetilde\Gamma(q;k), \label{eq:G2_gamma}
\end{align}
where we define function $\widetilde\Gamma (q;k)$ for $k=2,3,\ldots,$ as 
\begin{align}
\widetilde\Gamma(q;k)\equiv &\sum_{m=0}^{k-2}\binom{k-1}{m}(k-1-m)q^{m-1}(1-q)^{k-3-m} \left[ m(1-q)-(k-2-m)q\right] \notag \\
 & \times [F(m+1,k)-F(m,k)].
\end{align}
Note that $\widetilde\Gamma(q;2)=0$. 
As in Eq.~\eqref{eq:Gamma_arrange}, we rearrange $\widetilde\Gamma$ as
\begin{align}
    \widetilde\Gamma(q;k) =& \;\;\;\; \binom{k-1}{0}(k-1)q^{-1}(1-q)^{k-3}\left[0 - (k-2)q\right][F(1,k)-F(0,k)] \notag \\
     & + \binom{k-1}{1}(k-2)q^{0}(1-q)^{k-4}\left[1(1-q) - (k-3)q\right][F(2,k)-F(1,k)] \notag \\
     & + \binom{k-1}{2}(k-3)q^{1}(1-q)^{k-5}\left[2(1-q) - (k-4)q\right][F(3,k)-F(2,k)] \notag \\
     & + \cdots \notag \\
     & + \binom{k-1}{k-2}1\cdot q^{k-3}(1-q)^{-1}\left[(k-2)(1-q) - 0\right][F(k-1,k)-F(k-2,k)], \notag\\
     &\notag \\
     =& \;\;\;\;\frac{(k-1)!}{0!(k-3)!}q^{0}(1-q)^{k-3}\left[F(2,k)-2F(1,k)+F(0,k)\right] \notag\\
    & +\frac{(k-1)!}{1!(k-4)!}q^{1}(1-q)^{k-4}\left[F(3,k)-2F(2,k)+F(1,k)\right] \notag \\
   & + \frac{(k-1)!}{2!(k-5)!}q^{2}(1-q)^{k-5}\left[F(4,k)-2F(3,k)+F(2,k)\right] \notag \\
   & + \cdots \notag \\
   & + \frac{(k-1)!}{(k-3)!0!}q^{k-3}(1-q)^{0}\left[F(k-1,k)-2F(k-2,k)+F(k-3,k)\right]. 
    \label{eq:Gamma2_arrange}
\end{align}
It follows that 
\begin{align}
    \widetilde\Gamma(q;k) =& \sum_{s=0}^{k-3}\binom{k-1}{s}(k-1-s)(k-2-s)q^{s}(1-q)^{k-3-s} \notag \\
    & \times\left[F(s+2,k)-2F(s+1,k)+F(s,k)\right]. 
    \label{eq:Gamma_second}
\end{align}
Since $G^{''}(q)=(1-\rho_0)\sum_{k=3}^\infty \frac{k}{z}p_{k}\widetilde\Gamma(q;k)$ and $S^{''}(q)=(1-\rho_0)^{-1}G^{''}(q)$, we have Eq.~\eqref{eq:S_second}.
\bigskip

\add{
\section{Scale-free networks}\label{sec:scale_free}

In the current model, the degree distribution $\{p_k\}$ can take any functional form as long as the network is sparse and has a locally tree-like structure.
In the main text, we consider the standard \ER model, so $p_k$ is given by a Poisson distribution in the limit of large $N$.
On the other hand, it has been argued by a number of studies that real-world networks are \emph{scale-free}: the degrees of the networks follow power-law distributions where $p_k\sim k^{-\eta}$~\citep{Barabasi2016book,newman2018book2nd}.

As in Fig.~\ref{fig:phase_transition_monoplex}, Fig.~\ref{fig:vs_z_SF} shows a comparison between the two methods based on scale-free degree distributions (the exponent $\eta$ is set at $2.5$).
To generate scale-free networks with different average degrees, the prespecified minimum degree $k_{\rm min}$ is shifted from 1 to 6 ($z=0$ corresponds to empty network).
We find that the overall property is similar to the one we see in Fig.~\ref{fig:phase_transition_monoplex}; there is a non-negligible difference between the message-passing and mean-field methods in terms of the obtained $\rho^\ast$ while the range in which cascade occurs is slightly different from the one for \ER random graphs.

}

\section{Jacobian matrix $\mathcal{J}(q^A,q^B)$}\label{sec:jacobian_appendix}

Each element of the Jacobian matrix $\mathcal{J}(q^A,q^B)$ (Eq.~\ref{eq:Jacobian}) is given as follows:
\begin{align}
    {\J}_{11}(q^A,q^B) &= (1-\rho_0)\sum_{k_B=0}^\infty p_{k_B}\sum_{m_B=0}^{k_{B}}\mathcal{B}_{m_B}^{k_B}\left(q^{B}\right)\sum_{k_A=2}^\infty\frac{k_A p_{k_A}}{z_A} \sum_{m_A=0}^{k_A-2}(k_A-1-m_A)\binom{k_A-1}{m_A} \notag \\
    &\hspace{.7cm} \times 
    (q^A)^{m_A}(1-q^A)^{k_A-2-m_A}\left[ \mathcal{F}(m_A+1,m_B,\vect{k}) - \mathcal{F}(m_A,m_B,\vect{k})\right],\\
    \J_{12}(q^A,q^B) &= (1-\rho_0)\sum_{k_A=1}^\infty \frac{k_A p_{k_A}}{z_A}\sum_{m_A=0}^{k_{A}-1}\mathcal{B}_{m_A}^{k_A-1}\left(q^{A}\right)\sum_{k_B=1}^\infty p_{k_B}  \sum_{m_B=0}^{k_B-1}(k_B-m_B) \binom{k_B}{m_B}\notag \\
    &\hspace{.7cm}\times (q^B)^{m_B}(1-q^B)^{k_B-1-m_B}\left[ \mathcal{F}(m_A,m_B+1,\vect{k}) - \mathcal{F}(m_A,m_B,\vect{k})\right],\\
    \J_{21}(q^A,q^B) &= (1-\rho_0)\sum_{k_B=1}^\infty \frac{k_B p_{k_B}}{z_B}\sum_{m_B=0}^{k_{B}-1}\mathcal{B}_{m_B}^{k_B-1}\left(q^{B}\right)\sum_{k_A=1}^\infty p_{k_A} \sum_{m_A=0}^{k_A-1}(k_A-m_A) \binom{k_A}{m_A}\notag \\ 
    &\hspace{.7cm}\times (q^A)^{m_A}(1-q^A)^{k_A-1-m_A}\left[ \mathcal{F}(m_A+1,m_B,\vect{k}) - \mathcal{F}(m_A,m_B,\vect{k})\right],\\
    \J_{22}(q^A,q^B) &= (1-\rho_0)\sum_{k_A=0}^\infty p_{k_A}\sum_{m_A=0}^{k_{A}}\mathcal{B}_{m_A}^{k_A}\left(q^{A}\right)\sum_{k_B=2}^\infty\frac{k_B p_{k_B}}{z_B} \sum_{m_B=0}^{k_B-2}(k_B-1-m_B)\binom{k_B-1}{m_B}\notag \\
    &\hspace{.7cm}\times (q^B)^{m_B}(1-q^B)^{k_B-2-m_B}\left[ \mathcal{F}(m_A,m_B+1,\vect{k}) - \mathcal{F}(m_A,m_B,\vect{k})\right].
\end{align}
Note that from the monotonicity of the response function $\mathcal{F}$, all the elements of $\mathcal{J}$ are non-negative for $0\leq q^A\leq 1$ and $0\leq q^B\leq 1$. This proves that $g^{(A)}$ and $g^{(B)}$ are \addr{increasing}.


\begin{thebibliography}{}
\addr{
\bibitem[\protect\citeauthoryear{Amir, Evstigneev, and Gama}{Amir
  et~al.}{2021}]{amir2021oligopoly}
Amir, R., I.~Evstigneev, and A.~Gama (2021).
\newblock Oligopoly with network effects: Firm-specific versus single network.
\newblock {\em Economic Theory\/}~{\em 71\/}, 1203--1230.
}

\bibitem[\protect\citeauthoryear{Ballester, Calv{\'o}-Armengol, and
  Zenou}{Ballester et~al.}{2006}]{ballester2006}
Ballester, C., A.~Calv{\'o}-Armengol, and Y.~Zenou (2006).
\newblock Who's who in networks. Wanted: The key player.
\newblock {\em Econometrica\/}~{\em 74\/}, 1403--1417.

\bibitem[\protect\citeauthoryear{Barab{\'{a}}si}{Barab{\'{a}}si}{2016}]{Barabasi2016book}
Barab{\'{a}}si, A.-l. (2016).
\newblock {\em {Network Science}}.
\newblock Cambridge: Cambridge University Press.

\addr{
\bibitem[\protect\citeauthoryear{Baranga}{Baranga}{1991}]{baranga1991contraction}
Baranga, A. (1991).
\newblock The contraction principle as a particular case of Kleene's fixed
  point theorem.
\newblock {\em Discrete Mathematics\/}~{\em 98\/}, 75--79.

\bibitem[\protect\citeauthoryear{Barbieri}{Barbieri}{2021}]{barbieri2021complementarity}
Barbieri, S. (2021).
\newblock Complementarity and information in collective action.
\newblock {\em Economic Theory\/}, \url{https://doi.org/10.1007/s00199-021-01394-1}.
}

\bibitem[\protect\citeauthoryear{Bethe}{Bethe}{1935}]{bethe1935statistical}
 Bethe, H.~A. (1935).
 \newblock{Statistical theory of superlattices}.
 \newblock{\em Proceedings of the Royal Society A}~{\em 150}, 
  552--575.


\bibitem[\protect\citeauthoryear{Bianconi}{Bianconi}{2013}]{Bianconi2013}
Bianconi, G. (2013).
\newblock {Statistical mechanics of multiplex networks: Entropy and overlap}.
\newblock {\em Physical Review E\/}~{\em 87}, 062806.

\bibitem[\protect\citeauthoryear{Bianconi}{Bianconi}{2018}]{bianconi2018book}
Bianconi, G. (2018).
\newblock {\em Multilayer Networks: Structure and Function}.
\newblock Oxford University Press.

\bibitem[\protect\citeauthoryear{Brummitt and Kobayashi}{Brummitt and
  Kobayashi}{2015}]{Brummitt2015PRE}
Brummitt, C.~D. and T.~Kobayashi (2015).
\newblock Cascades in multiplex financial networks with debts of different
  seniority.
\newblock {\em Physical Review E}~{\em 91}, 062813.

\bibitem[\protect\citeauthoryear{Brummitt, Lee, and Goh}{Brummitt
  et~al.}{2012}]{Brummitt2012_PRER}
Brummitt, C.~D., K.-M. Lee, and K.-I. Goh (2012).
\newblock {Multiplexity-facilitated cascades in networks}.
\newblock {\em Physical Review E}~{\em 85}, 045102(R).

\bibitem[\protect\citeauthoryear{Caccioli, Barucca, and Kobayashi}{Caccioli
  et~al.}{2018}]{caccioli2018review}
Caccioli, F., P.~Barucca, and T.~Kobayashi (2018).
\newblock Network models of financial systemic risk: a review.
\newblock {\em Journal of Computational Social Science\/}~{\em 1\/},
  81--114.

\bibitem[\protect\citeauthoryear{Chen, Zenou, and Zhou}{Chen
  et~al.}{2018}]{chen2018AEJmultiple}
Chen, Y.-J., Y.~Zenou, and J.~Zhou (2018).
\newblock Multiple activities in networks.
\newblock {\em American Economic Journal: Microeconomics\/}~{\em 10\/},
  34--85.

\bibitem[\protect\citeauthoryear{Cont, Moussa, and Santos}{Cont
  et~al.}{2013}]{Cont2013}
Cont, R., A.~Moussa, and E.~B. Santos (2013).
\newblock Network structure and systemic risk in banking systems.
\newblock In J.-P. Fouque and J.~A. Langsam (Eds.), {\em Handbook on Systemic
  Risk}. Cambridge University Press, New York.

\bibitem[\protect\citeauthoryear{Dall'Asta}{Dall'Asta}{2021}]{dall2021coordination}
Dall'Asta, Luca (2021).
\newblock {Coordination problems on networks revisited: Statics and dynamics}.
\newblock {\em arXiv:2106.02548}.

\bibitem[\protect\citeauthoryear{Dorogovtsev, Goltsev and Mendes}{Dorogovtsev
  et~al.}{2008}]{Dorogovtsev2008RevModPhys}
Dorogovtsev, S.~N., A.~V. Goltsev, and J.~F.~F. Mendes (2008).
\newblock {Critical phenomena in complex networks}.
\newblock {\em Reviews of Modern Physics}~{\em 80}, 1275.


\bibitem[\protect\citeauthoryear{Easley and Kleinberg}{Easley and
  Kleinberg}{2010}]{Easley2010book}
Easley, D. and J.~Kleinberg (2010).
\newblock {\em Networks, Crowds, and Markets: Reasoning about a Highly
  Connected World}.
\newblock Cambridge University Press, Cambridge.

\bibitem[\protect\citeauthoryear{Erd\H{o}s and R\'enyi}{Erd\H{o}s and
  R\'enyi}{1959}]{Erdos1959PublMath}
Erd\H{o}s, P. and A.~R\'enyi (1959).
\newblock On random graphs.
\newblock {\em Publicationes Mathematicae}~{\em 6}, 290--297.

\bibitem[\protect\citeauthoryear{Gai and Kapadia}{Gai and
  Kapadia}{2010}]{GaiKapadia2010}
Gai, P. and S.~Kapadia (2010).
\newblock {Contagion in financial networks}.
\newblock {\em Proceedings of the Royal Society A}~{\em 466}, 2401--2423.


\bibitem[\protect\citeauthoryear{Galeotti, Goyal, Jackson, Vega-Redondo, and
  Yariv}{Galeotti et~al.}{2010}]{galeotti2010network}
Galeotti, A., S.~Goyal, M.~O. Jackson, F.~Vega-Redondo, and L.~Yariv (2010).
\newblock Network games.
\newblock {\em Review of Economic Studies\/}~{\em 77\/}, 218--244.

\bibitem[\protect\citeauthoryear{Gleeson and Cahalane}{Gleeson and
  Cahalane}{2007}]{Gleeson2007}
Gleeson, J. and D.~Cahalane (2007).
\newblock {Seed size strongly affects cascades on random networks}.
\newblock {\em Physical Review E\/}~{\em 75\/}, 56103.

\bibitem[\protect\citeauthoryear{Gleeson}{Gleeson}{2008}]{Gleeson2008}
Gleeson, J.~P. (2008).
\newblock {Cascades on correlated and modular random networks}.
\newblock {\em Physical Review E\/}~{\em 77\/}, 46117.

\bibitem[\protect\citeauthoryear{Gleeson}{Gleeson}{2011}]{gleeson2011high}
Gleeson, J.~P. (2011).
\newblock High-accuracy approximation of binary-state dynamics on networks.
\newblock {\em Physical Review Letters\/}~{\em 107\/}, 068701.

\bibitem[\protect\citeauthoryear{Gleeson}{Gleeson}{2013}]{gleeson2013binary}
Gleeson, J.~P. (2013).
\newblock Binary-state dynamics on complex networks: Pair approximation and
  beyond.
\newblock {\em Physical Review X\/}~{\em 3\/}, 021004.

\bibitem[\protect\citeauthoryear{Gleeson and Porter}{Gleeson and
  Porter}{2018}]{gleeson2018message}
Gleeson, J.~P. and M.~A. Porter (2018).
\newblock Message-passing methods for complex contagions.
\newblock In S. Lehmann and Y.-Y. Ahn (Eds.),
 {\em Complex Spreading Phenomena in Social Systems}, 81--95.
  Springer.

\bibitem[\protect\citeauthoryear{Goyal and Janssen}{Goyal and
  Janssen}{1997}]{goyal1997non}
Goyal, S. and M.~C. Janssen (1997).
\newblock Non-exclusive conventions and social coordination.
\newblock {\em Journal of Economic Theory}~{\em 77\/}, 34--57.

\bibitem[\protect\citeauthoryear{Granovetter}{Granovetter}{1978}]{Granovetter1978}
Granovetter, M. (1978).
\newblock {Threshold models of collective behavior}.
\newblock {\em American Journal of Sociology\/}~{\em 83\/}, 1420--1443.

\bibitem[\protect\citeauthoryear{Hurd}{Hurd}{2016}]{hurd2016Book}
Hurd, T.~R. (2016).
\newblock {\em Contagion!: Systemic Risk in Financial Networks}.
\newblock Springer, Berlin.

\bibitem[\protect\citeauthoryear{Ikeda, Hasegawa, and Nemoto}{Ikeda
  et~al.}{2010}]{ikeda2010cactus}
Ikeda, Y., T.~Hasegawa, and K.~Nemoto (2010).
\newblock Cascade dynamics on clustered network.
\newblock {\em Journal of Physics: Conference Series} 221,
  012005.

\bibitem[\protect\citeauthoryear{Immorlica, Kleinberg, Mahdian, and
  Wexler}{Immorlica et~al.}{2007}]{immorlica2007role}
Immorlica, N., J.~Kleinberg, M.~Mahdian, and T.~Wexler (2007).
\newblock The role of compatibility in the diffusion of technologies through
  social networks.
\newblock In {\em Proceedings of the 8th ACM Conference on Electronic
  Commerce}, 75--83.

\bibitem[\protect\citeauthoryear{Jackson}{Jackson}{2008}]{Jackson2008book}
Jackson, M.~O. (2008).
\newblock {\em Social and Economic Networks}.
\newblock Princeton University Press, Princeton.

\bibitem[\protect\citeauthoryear{Jackson}{Jackson}{2011}]{jackson2011overview}
Jackson, M.~O. (2011).
\newblock An overview of social networks and economic applications.
\newblock In J. Benhabib, A. Bisin, and M.~O. Jackson, (eds.) {\em Handbook of Social Economics} 1, 511--585.
  Elsevier.

\bibitem[\protect\citeauthoryear{Jackson, Rogers, and Zenou}{Jackson
  et~al.}{2017}]{jacksonJEL2017}
Jackson, M.~O., B.~W. Rogers, and Y.~Zenou (2017).
\newblock The economic consequences of social-network structure.
\newblock {\em Journal of Economic Literature\/}~{\em 55\/}, 49--95.

\bibitem[\protect\citeauthoryear{Jackson and Yariv}{Jackson and
  Yariv}{2007}]{jackson2007diffusion}
Jackson, M.~O. and L.~Yariv (2007).
\newblock Diffusion of behavior and equilibrium properties in network games.
\newblock {\em American Economic Review\/}~{\em 97\/}, 92--98.

\bibitem[\protect\citeauthoryear{Jackson and Zenou}{Jackson and
  Zenou}{2015}]{jackson-zenou2015games}
Jackson, M.~O. and Y.~Zenou (2015).
\newblock Games on networks.
\newblock In {\em Handbook of Game Theory with Economic Applications}
  4, 95--163. Elsevier.

\addr{
\bibitem[\protect\citeauthoryear{Kamihigashi, Reffett, and Yao}{Kamihigashi
  et~al.}{2015}]{kamihigashi2015application}
Kamihigashi, T., K.~Reffett, and M.~Yao (2015).
\newblock An application of Kleene's fixed point theorem to dynamic
  programming.
\newblock {\em International Journal of Economic Theory\/}~{\em 11\/},
  429--434.
}

\bibitem[\protect\citeauthoryear{Kandori, Mailath, and Rob}{Kandori
  et~al.}{1993}]{kandori1993ECMAlearning}
Kandori, M., G.~J. Mailath, and R.~Rob (1993).
\newblock Learning, mutation, and long run equilibria in games.
\newblock {\em Econometrica\/}~{\em 61}, 29--56.

\bibitem[\protect\citeauthoryear{Karimi and Holme}{Karimi and
  Holme}{2013}]{Karimi2013PhysicaA}
Karimi, F. and P.~Holme (2013).
\newblock {Threshold model of cascades in empirical temporal networks}.
\newblock {\em Physica A\/}~{\em 392}, 3476--3483.

\bibitem[\protect\citeauthoryear{Kobayashi}{Kobayashi}{2015}]{kobayashi2015trend}
Kobayashi, T. (2015).
\newblock Trend-driven information cascades on random networks.
\newblock {\em Physical Review E\/}~{\em 92\/}(6), 062823.

\bibitem[\protect\citeauthoryear{Kochenberger, Hao, Glover, Lewis, L{\"u},
  Wang, and Wang}{Kochenberger et~al.}{2014}]{kochenberger2014qubo}
Kochenberger, G., J.-K. Hao, F.~Glover, M.~Lewis, Z.~L{\"u}, H.~Wang, and
  Y.~Wang (2014).
\newblock The unconstrained binary quadratic programming problem: A survey.
\newblock {\em Journal of Combinatorial Optimization\/}~{\em 28\/}, 58--81.

\addr{
\bibitem[\protect\citeauthoryear{La~Torre, Liuzzi, and Marsiglio}{La~Torre
  et~al.}{2022}]{la2022geographical}
La~Torre, D., D.~Liuzzi, and S.~Marsiglio (2022).
\newblock Geographical heterogeneities and externalities in an
  epidemiological-macroeconomic framework.
\newblock {\em Journal of Public Economic Theory\/}, forthcoming.
}

\bibitem[\protect\citeauthoryear{Lelarge}{Lelarge}{2012}]{lelarge2012diffusion}
Lelarge, M. (2012).
\newblock Diffusion and cascading behavior in random networks.
\newblock {\em Games and Economic Behavior\/}~{\em 75\/}, 752--775.

\bibitem[\protect\citeauthoryear{L{\'o}pez-Pintado}{L{\'o}pez-Pintado}{2006}]{lopez2006contagionIJGT}
L{\'o}pez-Pintado, D. (2006).
\newblock Contagion and coordination in random networks.
\newblock {\em International Journal of Game Theory\/}~{\em 34\/}, 371--381.

\bibitem[\protect\citeauthoryear{L{\'o}pez-Pintado}{L{\'o}pez-Pintado}{2008}]{lopez2008GEBdiffusion}
L{\'o}pez-Pintado, D. (2008).
\newblock Diffusion in complex social networks.
\newblock {\em Games and Economic Behavior\/}~{\em 62\/}, 573--590.

\bibitem[\protect\citeauthoryear{L{\'o}pez-Pintado}{L{\'o}pez-Pintado}{2012}]{lopez2012GEBinfluence}
L{\'o}pez-Pintado, D. (2012).
\newblock Influence networks.
\newblock {\em Games and Economic Behavior\/}~{\em 75\/}, 776--787.

\addr{
\bibitem[\protect\citeauthoryear{Luo, Mauleon, and Vannetelbosch}{Luo
  et~al.}{2021}]{luo2021network}
Luo, C., A.~Mauleon, and V.~Vannetelbosch (2021).
\newblock Network formation with myopic and farsighted players.
\newblock {\em Economic Theory\/}~{\em 71\/}, 1283--1317.

\bibitem[\protect\citeauthoryear{Masatlioglu and Suleymanov}{Masatlioglu and
  Suleymanov}{2021}]{masatlioglu2021decision}
Masatlioglu, Y. and E.~Suleymanov (2021).
\newblock Decision making within a product network.
\newblock {\em Economic Theory\/}~{\em 71\/}, 185--209.
}


\bibitem[\protect\citeauthoryear{Melnik, Hackett, Porter, Mucha, and
  Gleeson}{Melnik et~al.}{2011}]{melnik2011unreasonable}
Melnik, S., A.~Hackett, M.~A. Porter, P.~J. Mucha, and J.~P. Gleeson (2011).
\newblock The unreasonable effectiveness of tree-based theory for networks with
  clustering.
\newblock {\em Physical Review E\/}~{\em 83\/}, 036112.

\bibitem[\protect\citeauthoryear{Melo, Emerson}{Melo}{2021}]{melo2021uniqueness}
Melo, E. (2021).
\newblock On the uniqueness of quantal response equilibria and its application to network games.
\newblock {\em Economic Theory\/}. \url{https://doi.org/10.1007/s00199-021-01385-2}.


\bibitem[\protect\citeauthoryear{Meng, Dawen and Tian, Guoqiang}{Meng and Tian}{2021}]{meng2021competitive}
Meng, D.\ and G.\ Tian (2021).
\newblock The competitive and welfare effects of long-term contracts with network externalities and bounded rationality.
\newblock {\em Economic Theory\/}~{\em 72\/}, 337--375.


\bibitem[\protect\citeauthoryear{Mezard et~al.}{M\'ezard et~al.}{1987}]{mezard1987spin}
M{\'e}zard, M., G. Parisi, and M.~A. Virasoro (1987).
\newblock {\em Spin Glass Theory and Beyond: An Introduction to the Replica Method and Its Applications}.
\newblock World Scientific Publishing Company.



\bibitem[\protect\citeauthoryear{Molloy and Reed}{Molloy and
  Reed}{1995}]{molloy1995critical}
Molloy, M. and B.~Reed (1995).
\newblock A critical point for random graphs with a given degree sequence.
\newblock {\em Random Structures \& Algorithms\/}~{\em 6\/}, 161--180.

\bibitem[\protect\citeauthoryear{Morris}{Morris}{2000}]{morris2000contagion}
Morris, S. (2000).
\newblock Contagion.
\newblock {\em Review of Economic Studies}~{\em 67\/}, 57--78.

\bibitem[\protect\citeauthoryear{Nematzadeh, Ferrara, Flammini, and
  Ahn}{Nematzadeh et~al.}{2014}]{Nematzadeh2014}
Nematzadeh, A., E.~Ferrara, A.~Flammini, and Y.-Y. Ahn (2014).
\newblock Optimal network modularity for information diffusion.
\newblock {\em Physical Review Letters}~{\em 113}, 088701.

\bibitem[\protect\citeauthoryear{Newman}{Newman}{2018}]{newman2018book2nd}
Newman, M. (2018).
\newblock {\em Networks, 2nd ed.}
\newblock Oxford University Press, Oxford.

\bibitem[\protect\citeauthoryear{Oyama and Takahashi}{Oyama and
  Takahashi}{2015}]{oyama2015bilingual}
Oyama, D. and S.~Takahashi (2015).
\newblock Contagion and uninvadability in local interaction games: The
  bilingual game and general supermodular games.
\newblock {\em Journal of Economic Theory\/}~{\em 157}, 100--127.


\bibitem[\protect\citeauthoryear{Ruan}{Ruan}{2015}]{Ruan2015PRL}
Ruan, Z., G.\ Iniguez, M.~Karsai and J.~Kert{\'e}sz (2015).
\newblock Kinetics of social contagion.
\newblock {\em Physical Review Letters\/}~{\em 115\/}, 218702.

\bibitem[\protect\citeauthoryear{Sadler}{Sadler}{2020}]{sadler2020diffusion}
Sadler, E. (2020).
\newblock Diffusion games.
\newblock {\em American Economic Review\/}~{\em 110\/}, 225--70.

\addr{
\bibitem[\protect\citeauthoryear{Stoltenberg-Hansen, Lindstr{\"o}m, and
  Griffor}{Stoltenberg-Hansen et~al.}{1994}]{stoltenberg1994mathematical}
Stoltenberg-Hansen, V., I.~Lindstr{\"o}m, and E.~R. Griffor (1994).
\newblock {\em Mathematical Theory of Domains}.
\newblock No.~22. Cambridge University Press.

\bibitem[\protect\citeauthoryear{Tabasso}{Tabasso}{2019}]{tabasso2019diffusion}
Tabasso, N. (2019).
\newblock Diffusion of multiple information: On information resilience and the
  power of segregation.
\newblock {\em Games and Economic Behavior\/}~{\em 118}, 219--240.
}

\bibitem[\protect\citeauthoryear{Unicomb, I{\~n}iguez, Gleeson, and
  Karsai}{Unicomb et~al.}{2021}]{unicomb2021dynamics}
Unicomb, S., G.~I{\~n}iguez, J.~P. Gleeson, and M.~Karsai (2021).
\newblock Dynamics of cascades on burstiness-controlled temporal networks.
\newblock {\em Nature Communications}~{\em 12\/}, 1--10.

\bibitem[\protect\citeauthoryear{Unicomb, I{\~n}iguez, Kert{\'e}sz, and
  Karsai}{Unicomb et~al.}{2019}]{unicomb2019reentrant}
Unicomb, S., G.~I{\~n}iguez, J.~Kert{\'e}sz, and M.~Karsai (2019).
\newblock Reentrant phase transitions in threshold driven contagion on
  multiplex networks.
\newblock {\em Physical Review E\/}~{\em 100\/}, 040301.

\bibitem[\protect\citeauthoryear{Watts}{Watts}{2002}]{Watts2002}
Watts, D.~J. (2002).
\newblock A simple model of global cascades on random networks.
\newblock {\em Proceedings of the National Academy of Sciences of the United States of America}~{\em 99\/}, 5766--5771.

\bibitem[\protect\citeauthoryear{Watts and Dodds}{Watts and
  Dodds}{2007}]{Watts2007}
Watts, D.~J. and P.~S. Dodds (2007).
\newblock {Influentials, networks, and public opinion formation}.
\newblock {\em Journal of Consumer Research}~{\em 34\/}, 441--458.

\bibitem[\protect\citeauthoryear{Weiss}{Weiss}{1907}]{weiss1907}
Weiss, P. (1907).
\newblock {L'hypothèse du champ mol\'{e}culaire et la propri\'{e}t\'{e} ferromagn\'{e}tique}.
\newblock {\em Journal de Physique Th\'{e}orique et Appliqu\'{e}e}~{\em 6\/}, 661--690.

\bibitem[\protect\citeauthoryear{Ya\u{g}an and Gligor}{Ya\u{g}an and
  Gligor}{2012}]{Yagan2012}
Ya\u{g}an, O. and V.~Gligor (2012).
\newblock {Analysis of complex contagions in random multiplex networks}.
\newblock {\em Physical Review E}~{\em 86\/}, 036103.


\end{thebibliography}


\clearpage

\setcounter{section}{0}
\setcounter{table}{0}
\setcounter{equation}{0}
\setcounter{figure}{0}
\setcounter{page}{1}
     
\renewcommand{\thetable}{S\arabic{table}}
\renewcommand{\thefigure}{S\arabic{figure}}
\renewcommand{\thesection}{S\arabic{section}}
\renewcommand{\theequation}{S\arabic{equation}}

{\flushleft
{\fontsize{16pt}{16pt}\selectfont
 \textbf{Online Appendix} \\
 \vspace{.7cm}
 \Large{``Dynamics of diffusion on monoplex and multiplex networks: \\ A message-passing approach"} \\
 \vspace{.5cm}
{\large Teruyoshi Kobayashi and Tomokatsu Onaga}
}
}
\vspace{2cm}
\begin{figure}[hb]
     \centering
     \includegraphics[width=12cm]{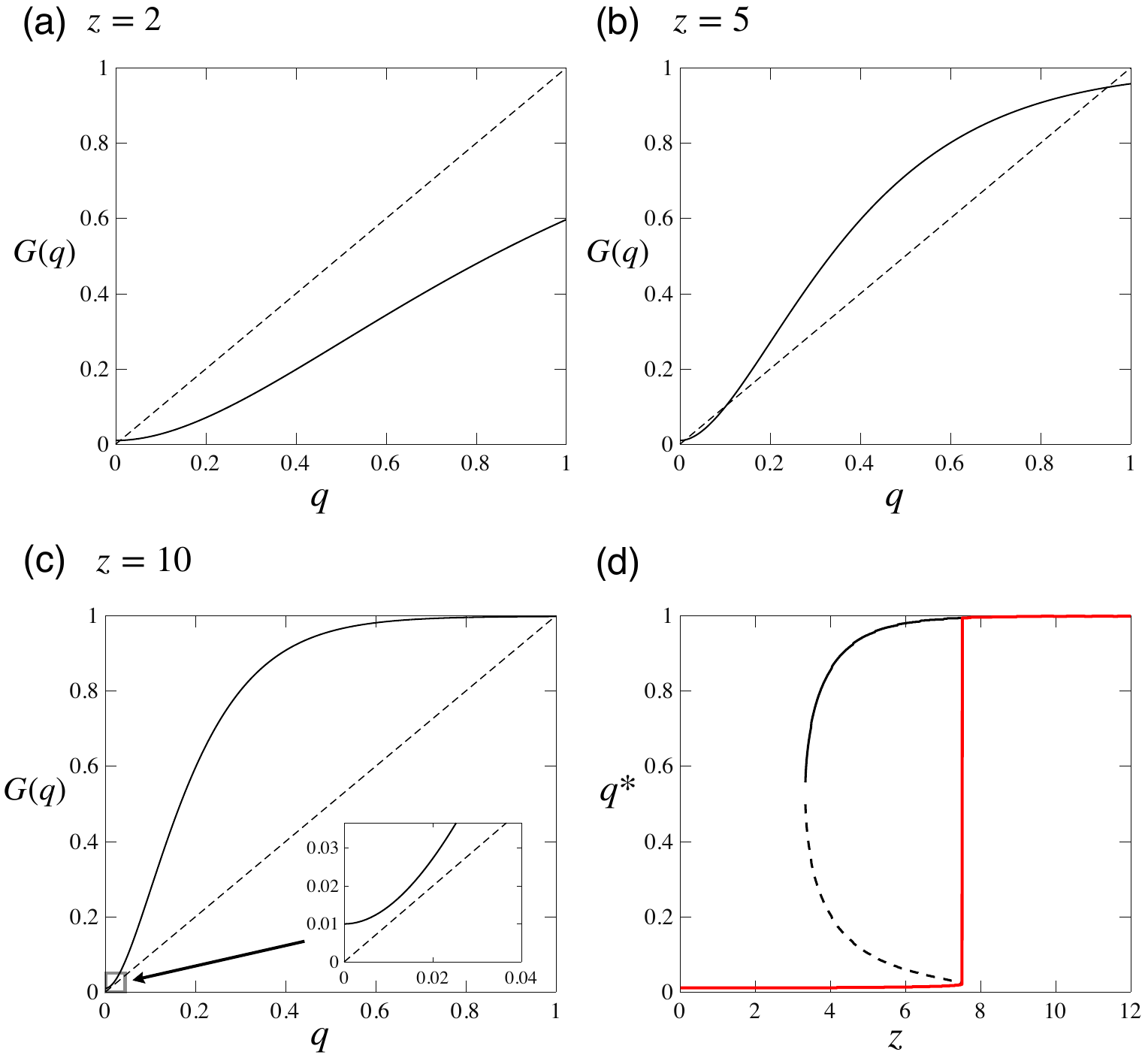}
     \caption{Fixed points of recursion equation \eqref{eq:recursion_q} for \new{the absolute threshold model}. $\rho_0=0.01$ and $\theta=1.5$. The mean degree $z$ is (a) $2$, (b) $5$ and (c) 10. See the caption of Fig.~\ref{fig:bifurcation_Watts} for details.}
     \label{fig:bifurcation_Grano}
 \end{figure}

\begin{figure}[tbh]
    \centering
    \includegraphics[width=13cm]{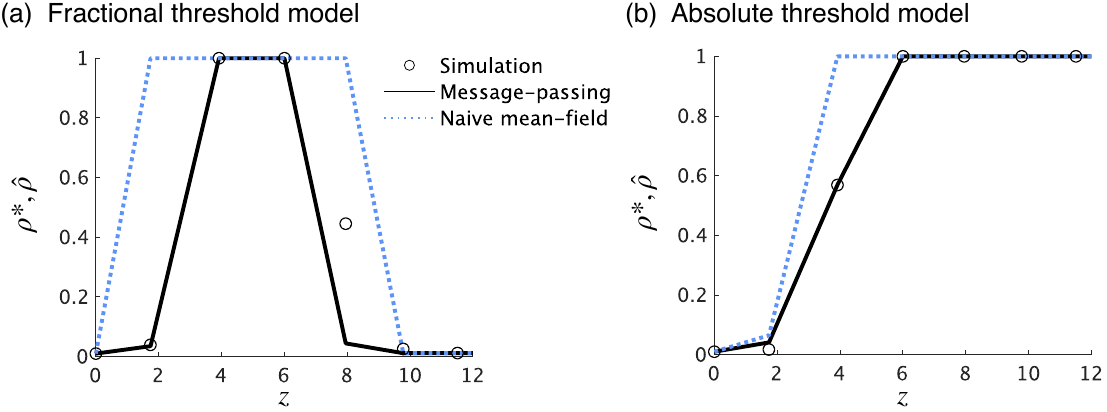}
    \caption{\add{Steady-state fraction of active players in theory ($\rho^*$) and numerical experiments ($\hat\rho$) based on scale-free networks.
    The minimum degree is varied from $1$ to $6$ to generate networks with different average degrees. 
    The exponent of the scale-free degree distributions is set at $2.5$.
    See the caption of Fig.~\ref{fig:phase_transition_monoplex} for the other parameter values.
    }}
    \label{fig:vs_z_SF}
\end{figure}

\begin{figure}[thb]
     \centering
     \includegraphics[width=15cm]{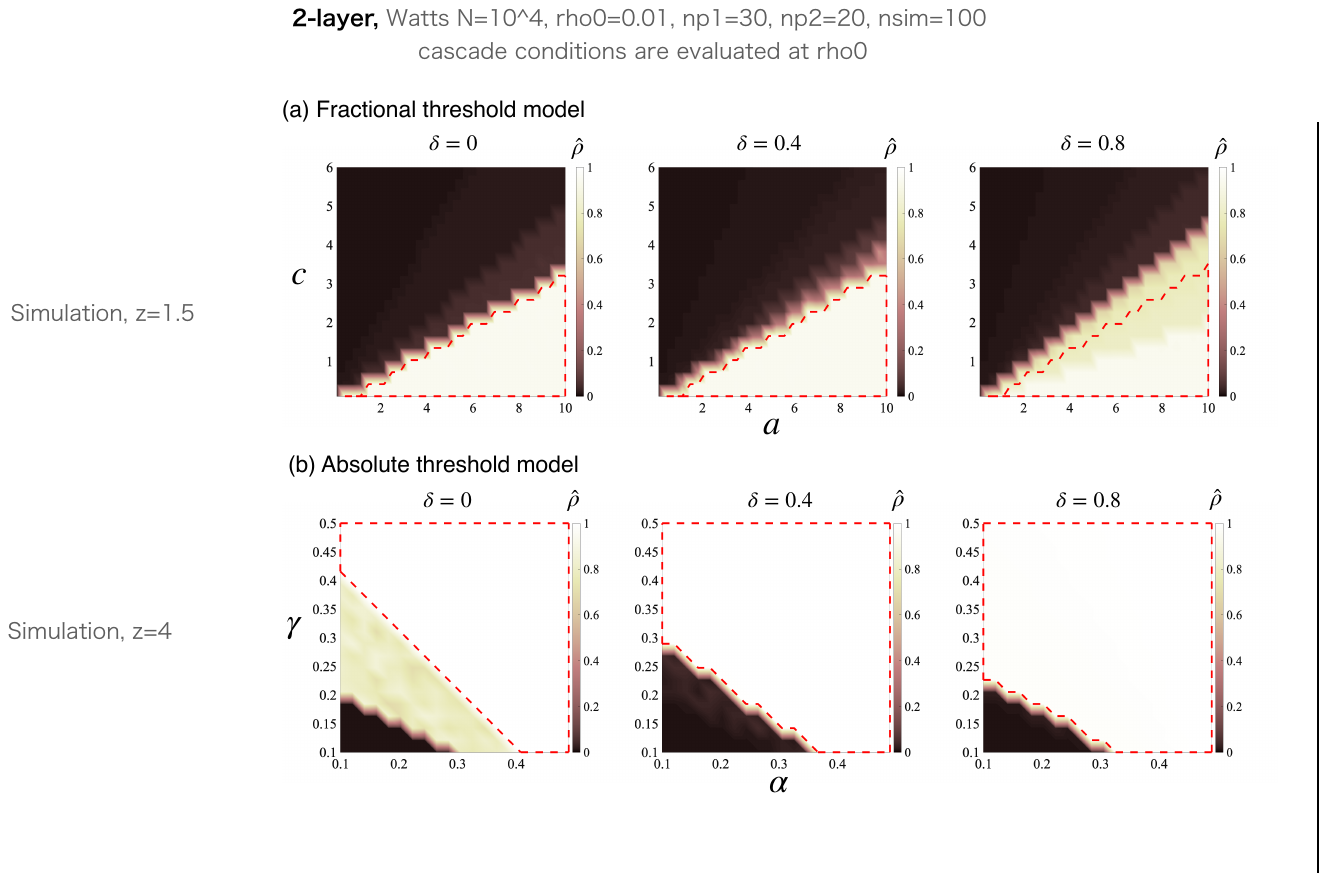}
     \caption{Simulated and theoretical cascade regions in two-layer multiplex networks. $\delta$ denotes the degree of inter-layer heterogeneity in the threshold condition. $z=1.5$ and $4$ for (a) \new{the fractional threshold model} and (b) \new{absolute threshold model}, respectively. The case of no inter-layer heterogeneity (i.e., $\delta=0$) corresponds to the middle column of Fig.~\ref{fig:colormap_watts_grano} since the two layers can be aggregated. $N=10^4$ and $\rho_0=0.01$.}
     \label{fig:colormap_multiplex}
 \end{figure}


\end{CJK*}

\end{document}